\documentclass[11pt]{article}
\usepackage{color, colortbl}
\usepackage{verbatim}
\usepackage{siunitx} 
\definecolor{Gray}{gray}{0.9}
\definecolor{White}{gray}{1}

\def\ve{\varepsilon}

\newcommand{\mcitet}[1]{\mbox{\citet{#1}}}

\usepackage{amsthm}

\newtheorem{proposition}{Proposition}

\newtheorem{lemma}{Lemma}
\newtheorem{corollary}{Corollary}

\newtheorem{assumption}{Assumption}
\usepackage{longtable,rotating}
\usepackage{graphics, amssymb, amsmath, setspace, color}
\usepackage{verbatim}
\usepackage{ulem}  
\usepackage{url}
\usepackage[round]{natbib}

\usepackage[title]{appendix}

\usepackage{hyperref}
\hypersetup{
colorlinks,
linkcolor=blue, citecolor=blue}

\begin{document}
\title{\vspace{-0.1 in}Inference in Predictive Quantile Regressions\thanks{We are grateful to an anonymous referee for comments that led to substantial improvements. We thank Ji Hyung Lee for his helpful and detailed comments on an early draft. We thank Ying Chen, Chuan Goh, Bruce Hansen, Christian Gourieroux, Frank Kleibergen, Benoit Perron, Chi Wan, Zhijie  Xiao,  conference participants at the Frontiers in Theoretical Econometrics, the European Meetings of the Econometric Society, Joint Statistical Meetings, the CIREQ Time Series Conference, the Canadian Econometrics Study Group, the Midwest Econometrics Study Group, the SMU-ESSEC Symposium on Empirical Finance and Financial Econometrics, and the Canadian Economic Association and seminar participants at the University of Guelph, Ryerson University, the University of Waterloo, Nanjing University, Zhejiang University, Shandong University, Hitotsubashi University, and York University for useful comments and discussion. We thank Amit Goyal for the use of his publicly posted data and his helpful answers to several queries. Maynard and Shimotsu thank the SSHRC for research funding under grant number 410-2010-0074. We take full responsibility for any remaining errors.}}

\author{Alex Maynard\thanks{Department of Economics, University of Guelph, 50 Stone Road East,
Guelph, Ontario N1G 2W1, Canada. E-mail: maynarda@uoguelph.ca}\\ University of Guelph
\and Katsumi Shimotsu\thanks{Faculty of Economics, University of Tokyo,
7-3-1, Hongo, Bunkyo-ku, Tokyo 113-0033, Japan. E-mail: shimotsu@e.u-tokyo.ac.jp}\\ University of Tokyo
\and Nina Kuriyama\thanks{formerly Yini Wang (maiden name).  Department of Economics, Renmin University.
E-mail: wangyini@ruc.edu.cn}\\ Renmin University 
} 
\date{May 4, 2024}
\maketitle
\thispagestyle{empty}

\vspace{-0.2 in}
\begin{abstract}\addcontentsline{toc}{section}{Abstract}
\noindent This paper studies inference in predictive quantile regressions when the predictive regressor has a near-unit root. We derive asymptotic distributions for the quantile regression estimator and its heteroskedasticity and autocorrelation consistent (HAC) $t$-statistic in terms of functionals of Ornstein-Uhlenbeck processes. We then propose a switching-fully modified (FM) predictive test for quantile predictability. The proposed test employs an FM style correction with a Bonferroni bound for the local-to-unity parameter when the predictor has a near unit root. It switches to a standard predictive quantile regression test with a slightly conservative critical value when the largest root of the predictor lies in the stationary range.  Simulations indicate that the test has a reliable size in small samples and good power. We employ this new methodology to test the ability of three commonly employed, highly persistent and endogenous lagged valuation regressors -- the dividend price ratio, earnings price ratio, and book-to-market ratio -- to predict the median, shoulders, and tails of the stock return distribution. \\


\end{abstract}


\newpage


\section{Introduction}

In this paper, we develop asymptotic theory in the context of predictive quantile regressions with nearly integrated regressors. Beginning with influential work by \citet{Shiller84}, \citet{Campbell&Shiller88a,Campbell&Shiller88b}, \citet{Fama&French88} and \citet{Hod92},  there has been extensive literature on testing whether a variety of proposed predictors can forecast mean stock returns. This has implications, not only for the risk neutral market efficiency hypothesis, but also for portfolio analysis. Indeed, subsequent empirical work debates the ability of investors to use predictors, such as dividend or earning price ratios, to create dynamic asset allocation strategies that outperform the market \citep{Goyal:Welch:2008,Campbell&Thompson2008}.

While most empirical literature has focused exclusively on predictive means or variances, the portfolio decision often depends on the entire return distribution. Likewise, the tails of the distribution are of particular interest to risk managers and are also important to policymakers, who must consider the worst case, as well as baseline, forecast scenarios. \citet{Cenesizoglu:Timmermann:08} employ the quantile regression method introduced by \citet{KoenkerBassett78} to extract a richer set of return predictions. They find that a number of predictors have little information for the center of the distribution yet have important and often asymmetric implications for the tails.

One reason that the ongoing debate over predictive mean regression has lasted so long is that the limiting distribution of the standard $t$-statistic is nonstandard. Firstly, the predictor variables, such as dividend yields, dividend price and earning price ratios, are strongly autocorrelated. Secondly, although pre-determined, these predictors are not strictly exogenous because their innovations are often highly correlated with the error term in the predictive regression. Consequently, tests using the standard normal critical values will over-reject the null hypothesis of non-predictability, as is found by \citet{Mankiw&Shapiro86,Stambaugh86,Cavanagh/Elliott/Stock:95,stambaugh99}.

Much attention has been devoted to overcoming such size distortions in predictive mean regressions, resulting in a rich literature. Perhaps the most popular approach has been the use of an explicit local-to-unity specification for the predictor.\footnote{The literature on mean predictive tests is too extensive to provide a full review here. Other prominent approaches include the IVX approach  \citep{PhillipsMagdelliano2009wp,KostakisMagdalinosStamatogiannis:2015}, nearly \citep{ElliotMullerWatson15} and conditionally optimal tests \citep{JanssonMoreira06}, linear projection methods \citep{Cai14joe} and inference based on small sample distributions in parametric models \citep{Nelson&Kim93,stambaugh99,Lewellen04}, to name just a few.} 
\citet{Cavanagh/Elliott/Stock:95} propose corrected critical values based on a local-to-unity model with known values of the local-to-unity parameter ($c$). Since this parameter cannot be consistently estimated, they propose feasible inference methods using a Bonferroni bound and confidence interval on $c$ based on \citet{stock91}. \citet{cy06} develop an efficient test of predictability for a known local-to-unity parameter $c$. Since their correction depends on $c$,  a refined Bonferroni bounds procedure is employed for feasible inference. \citet{Hjalmarsson:07} notes that the \citet{cy06} procedure can be interpreted as a local-to-unity version of the fully modified estimator of \citet{Phillips&Hansen90}, and \citet{HjalmarssonJFE2011} proposes a generalization to long horizon returns.

In contrast to this large literature on predictive mean regression, we are aware of no theoretical work prior to our original working paper version \cite{MaynardShimotsuWang2011} that establishes valid econometric inference methods in quantile predictive regression with persistent regressors. In this paper, we develop proper inference methods for short-horizon predictive quantile regressions with nearly integrated regressors. This paper makes three main contributions. First, we derive the limit distribution of the quantile regression coefficients by generalizing results of \citet{xiao:09}, who derives inference in a quantile regression with cointegrated time series, to the local-to-unity setting. Second, we derive the asymptotics of heteroskedasticity and autocorrelation consistent (HAC) covariance matrix estimate and $t$-statistic. In contrast to predictive mean regression, the error terms in predictive quantile regression can be serially correlated. For example, when the stock return contains a GARCH component, the quantiles of the stock return are serially correlated because large returns are followed by large returns. Therefore, it is essential to use a HAC covariance matrix estimate and a HAC $t$-statistic. Existing literature in predictive quantile regression, such as  \citet{Lee2016}, \citet{FanLee19joe}, and \citet{cai23joe}, assume the error terms are serially uncorrelated. Consequently, their asymptotic results no longer hold, for example, when the stock return contains a GARCH component. As in the case of predictive mean regression, the limiting distribution of the standard and HAC $t$-statistics are nonstandard, and the standard inference procedures are unreliable when predictors are both persistent and endogenous. Third, we provide an inference procedure that is valid both when the predictor is a local-to-unity process and when the predictor is stationary. When the largest root of the predictor lies in the near unit root range, we provide a fully modified bias correction to the quantile regression estimator. This is equivalent to a quantile version of the correction in \citet{cy06}.  \citet{Phillips14} has proven that the predictive tests of \citet{Cavanagh/Elliott/Stock:95} and \citet{cy06} become invalid if the predictor is stationary. To address this problem, we follow in the spirit of \citet{ElliotMullerWatson15} and switch to a standard predictive quantile regression HAC $t$-test with a slightly conservative critical value when the largest root lies in the stationary range. We refer to this as a \textit{switching-FM predictive quantile regression test}. Our Monte Carlo simulations verify that the switching-FM quantile regression test has good size and power both when the predictor is a local-to-unity process and when the predictor is stationary. 

Subsequent to \citet{MaynardShimotsuWang2011}, \citet{Lee2016} develops the IVXQR test that uses a mildly integrated instrument generated by filtering the original predictor. Like our switching-FM test, the IVXQR test avoids the problems noted by \citet{Phillips14}. \citet{FanLee19joe} extend the IVXQR test to allow for heteroskedasticity and suggest a bootstrap inference.  Recently, \citet{cai23joe} develop a new test, $t^w$, that uses an auxiliary regressor formed by a weighted combination of an exogenous simulated nonstationary process and a bounded transformation of the original regressor. In \citet{cai23joe}'s simulation results, their $t^w$ test has better finite sample size and power than the IVXQR test. In our simulations, we find that the switching-FM test has higher power than the $t^w$ test with nearly comparable size, and the $t^w$ test has a modest size advantage at tail quantiles. On the other hand, the switching-FM test is designed for a single predictor. The tests of \citet{Lee2016}, \citet{FanLee19joe}, and \citet{cai23joe} have the distinctive advantage of generalizing easily to a multi-predictor setting. \cite{GungorLugo2019} develop a maximized Monte Carlo approach to exact finite sample inference in predictive quantile regression, but at the cost of requiring i.i.d.\ return innovations. 

The predictive quantile regression is more distantly related to sign, sign-rank, and directional tests of predictability \citep{Campbell&Dufour95,Campbell&Dufour97,GungorLugo2020} and to the (cross) quantilogram \citep{Linton:Whang:2007, HanLintonTaksushiWang16, LeeLintonWhang20}. More broadly, our results contribute to a rapidly developing literature in quantile regression for time series data. Although too numerous to survey here, developments include quantile autoregression \citep{KoenkerXiao2006,ChenKoenkerXiao2009}, dynamic quantile models \citep{EngleManganelli2004,GourierouxJasiak2008}, unit root quantile autoregression \citep{KoenkerXiao2004,Galvo2009}, and quantile cointegration \citep{xiao:09,cho15joe}.

The remainder of the paper is organized as follows. Section 2 establishes the framework of the problem and develops the asymptotic theory for predictive quantile regression under a local-to-unity specification.  In Section 3, a switching-FM predictive quantile test is proposed. In Section 4, results from our simulation study are reported. In Section 5, the techniques are applied to test the predictability of the stock return distribution using three commonly employed valuation predictors. Section 6 concludes the paper. The appendix provides proofs, and tables are included at the end.

In matters of notation, let $Q_{a_{t}}(\tau)$ and $Q_{a_{t}}(\tau|\mathcal{F})$ denote the unconditional and conditional $\tau$-quantile of $a_t$ conditional on $\mathcal{F}$. Let $\|A\|_r = (\sum_{ij}E|a_{ij}|^r)^{1/r}$ denote the $L^r$-norm. Let $:=$ denote ``equals by definition.'' Let $\Rightarrow$ denote weak convergence of the associated probability measures. Let $\equiv$ denote equality in distribution. Let $BM(\Omega)$ denote a Brownian motion with covariance matrix $\Omega$. Let $[x]$ denote the largest integer less than or equal to $x$. Let $MN(0,\Omega)$ denote the mixed normal distribution with variance $\Omega$. Continuous stochastic processes such as Brownian motion $B(r)$ on $[0,1]$ are usually written simply as $B$, and integrals $\int$ are understood to be taken over the interval $[0,1]$, unless specified otherwise. Let $I\{\cdot\}$ denote the indicator function. All limits below are taken as $T \to \infty$ unless stated otherwise.

\section{Predictive Quantile Regression}

\subsection{Model and Assumptions}\label{sec:model}

We model the conditional $\tau$-quantile of $y_t$ as
\begin{equation}\label{predictive-quantile-model}
Q_{y_t}(\tau|\mathcal{F}_{t-1}) = \gamma_0(\tau) + \gamma_1(\tau) x_{t-1} = \gamma(\tau)z_{t-1},
\end{equation}
where $y_t$ is typically a financial return,  $x_t$ is a predictor, such as earnings or dividend price ratio, $\mathcal{F}_{t-1}$ is the information contained in the lags of $x_t$, $\gamma(\tau):= (\gamma_0(\tau),\gamma_1(\tau))'$, and $z_t:= (1,x_t)'$. In model (\ref{predictive-quantile-model}), the dependence of $\gamma_{1}(\tau)$ on $\tau$ allows the impact of $x_{t-1}$ to vary across the quantiles of $y_t$. 

As noted by \citet{Cenesizoglu:Timmermann:08}, the predictive quantile model is robust to outliers and encompasses a number of other empirical models for financial returns.  For example, if $x_{t}$ is a variable with predictive content for volatility, such as squared returns or realized volatility, we may consider a model of the form $y_t = b_0 + b_1 x_{t-1} + (c_0 + c_1 x_{t-1} ) e_{t}$, where $e_t$ is independent of $\mathcal{F}_{t-1}$. The predictive quantile for $y_t$ then takes the form $Q_{y_t}(\tau | \mathcal{F}_{t-1}) = b_0 + c_0 Q_{e_t}(\tau) + (b_1 + c_1 Q_{e_t}(\tau)) x_{t-1}$. The predictive quantile model also encompasses a random-coefficient model \citep{KoenkerXiao2006}
\begin{equation}\label{model_randomcoeff}
y_t = b_0(e_t) + b_1(e_t) x_{t-1},
\end{equation}
where $e_t \sim U[0,1]$ is independent of $\mathcal{F}_{t-1}$. Provided that the right hand side is monotone increasing in $e_t$, the predictive quantile for $y_t$ in this model is $Q_{y_t}(\tau | \mathcal{F}_{t-1}) = b_0(\tau) + b_1 (\tau) x_{t-1}$.

Next, we consider the data-generating process for the predictor. Since most predictors employed in practice are highly persistent, we model the regressor $x_t$ as a near-unit root process. Specifically, we assume that
\begin{equation}\label{x}
x_t = (1 - \phi)x_{t-1} + v_t, \quad \phi =  1 + c/T, \quad 1\leq t \leq T,
\end{equation}
where $x_0 = o_p(T^{1/2})$, $v_t$ is a mean-zero stationary process, and $T$ is the sample size. A number of prior studies have used this framework to model predictors such as earnings and dividend price ratios, which are highly persistent but a priori stationary on economic grounds.

The standard quantile regression coefficient estimates are given by
\begin{equation} \label{qr_estimator}
(\widehat{\gamma}_0(\tau), \widehat{\gamma}_1(\tau)) := \arg\hspace{-0.13 in}\min_{\hspace{-0.12 in}(\gamma_0, \gamma_1) \in R^2}\sum_{t=1}^{T}\rho_{\tau}(y_t-\gamma_0-\gamma_1 x_{t-1}),
\end{equation}
where $\rho_\tau(u):=u\psi_\tau(u)$ with $\psi_\tau(u):=\tau-I(u<0)$ as in \citet{KoenkerBassett78}. When $\tau=0.5$, (\ref{qr_estimator}) gives the least absolute deviation estimator. Define  $\mathcal{F}_{t-1} := \sigma\{x_{t-j}, j\geq 1\}$,\label{F_defn} and
\begin{equation}\label{uttau_defn}
u_{t\tau}:= y_t - Q_{y_t}(\tau|\mathcal{F}_{t-1})= y_t - \gamma_0(\tau) - \gamma_1(\tau) x_{t-1},
\end{equation}
where the second equality follows from \eqref{predictive-quantile-model}. Since $\psi_\tau(u)=\tau-I(u<0)$, we have $E\left[\psi_\tau(u_{t\tau})\middle| \mathcal{F}_{t-1} \right]=0$ and $\text{var}\left[\psi_\tau(u_{t\tau})\middle| \mathcal{F}_{t-1} \right]=\tau(1-\tau)$.

In the literature, it has been recognized that the model $Q_{y_{t}}(\tau|\mathcal{F}_{t-1}) = \gamma_{0} + \gamma_{1}(\tau)x_{t-1}$ with quantile-varying $\gamma_{1}(\tau)$ poses difficulties for asymptotic analysis with local-to-unity regressors. When $x_{t-1}$ is a local-to-unity process and $\gamma_{1}(\tau)$ varies with $\tau$, $u_{t\tau}$ contains a local-to-unity component for some $\tau$ because, if $ \gamma_1(\tau_1) \neq  \gamma_1(\tau_2)$ for some $\tau_1 \neq \tau_2$, at least one of $u_{t\tau_1}=y_{t}-\gamma_0(\tau_1) - \gamma_1(\tau_1) x_{t-1}$ or $u_{t\tau_2} = y_{t}-\gamma_0(\tau_2) - \gamma_1(\tau_2) x_{t-1}$ contains a local-to-unity component. Consequently, the current literature assumes $\gamma_{1}(\tau)=\gamma_1$ for all $\tau \in (0,1)$ either explicitly or implicitly; see \citet[][Theorem 2.1]{Lee2016},  \citet{FanLee19joe}, and \citet{cai23joe}. 

In view of this, we explicitly impose $\gamma_{1}(\tau)=\gamma_1$ for all $\tau \in (0,1)$ as in \citet{xiao:09} and assume $u_{t\tau}$ is stationary in Assumption \ref{gamma-const} below. This rules out some interesting models, such as the random-coefficient model (\ref{model_randomcoeff}). We address this problem in Section \ref{sec:local_power} by considering models for which $\gamma_1(\tau)$ varies in $\tau$ but only locally.  Specifically, we will analyze the power of our tests under the model
\begin{equation*}
y_t = \gamma_0 + e_t + \gamma_1 x_{t-1}+ T^{\kappa-1}b(e_t) |x_{t-1}+\zeta|, \quad \kappa \in (0,1/2), 
\end{equation*}
where $e_t$ is independent of $\mathcal{F}_{t-1}$, $b(\cdot)$ is increasing, and $\zeta>0$ is non-random. In this model, the predictive quantile for $y_t$ is $Q_{y_t}(\tau | \mathcal{F}_{t-1}) = \gamma_0 (\tau) + \gamma_1 x_{t-1} + T^{\kappa-1}b (Q_{e_t}(\tau)) |x_{t-1}+\zeta|$, where $\gamma_0(\tau) = \gamma_0 + Q_{e_t}(\tau)$, and $x_{t-1}$ has a local quantile-varying effect, $T^{\kappa-1}b (Q_{e_t}(\tau)) |x_{t-1}+\zeta|$. Section \ref{sec:local_power} shows that our test statistic rejects $H_0: \gamma_{1}(\tau)=0$ with probability approaching one when $\gamma_1=0$ and $b(Q_{e_t}(\tau)) \neq 0$. In other words, our test can detect the existence of a local quantile-varying predictive component.

We collect the assumptions. Let $U_t(\tau) := (\psi_\tau(u_{t\tau}),v_t)'$. 

\begin{assumption}
\label{gamma-const}
$\gamma_1(\tau)=\gamma_1$ for all $\tau \in (0,1)$. Further, $u_{t\tau}$ is stationary.
\end{assumption}

\begin{assumption}
\label{error-dist1}
For each $\tau \in (0,1)$, the sequence of conditional stationary probability density functions $\{f_{u_{t \tau},t-1}(\cdot)\}$ of $\{u_{t\tau}\}$ given $\mathcal{F}_{t-1}$ is bounded above with probability one around zero, i.e., $f_{u_{t \tau},t-1}(\epsilon)< \infty$ with probability one for all $|\epsilon|< \eta$ for some $\eta>0$. Further, $f_{u_{ \tau}}(0) := E[f_{u_{t \tau},t-1}(0)]>0$.
\end{assumption}

\begin{assumption}
\label{mixing}
For each $\tau \in (0,1)$, $\{U_t(\tau), f_{u_{t \tau},t-1}(0)\}$ is a stationary strong mixing sequence with mixing coefficient $\alpha_m$ of size $-p\beta/(p-\beta)$ for some $p>\beta \geq 3$ and $\|U_t(\tau)\|_p < \infty$. Further, $\Omega_\tau := \sum_{h=-\infty}^{\infty} E[U_{t+h}(\tau) U_{t}(\tau)'] < \infty$.
\end{assumption}

Assumption \ref{error-dist1} is similar to Assumption 2.1(i) in \citet{Lee2016} and Assumption A2(i) in \citet{cai23joe}. Assumption \ref{mixing} is essentially the same as Assumption 1 in \citet{Hansen:92}. Note that the condition $\lim_{n\to\infty}n^{-1}E(V_{n}V_{n}') =\Omega<\infty$ in \citet{Hansen:92} is satisfied because we assume $U_t(\tau)$ is stationary. We assume $\beta \geq 3$ because our proof uses Theorem 4.2 of \citet{Hansen:92}.
\newcommand{\reportTwoAdditionalCommentTwo}{Assumption \ref{mixing} imposes the mixing condition on both $u_{t \tau}$ and directly on $f_{u_\tau, t-1}(0)$.}\label{report2:Additional:Comment2}\reportTwoAdditionalCommentTwo\ 
 For GARCH($p,q$) and augmented GARCH models, \citet{carrascochen02et} show conditions under which the squared residuals $u_t^2$ and the latent conditional volatility process $\sigma_{t+1}^2$ are jointly mixing. Then, $f_{u_\tau, t-1}(0)$ is also mixing because we can write the conditional density of $u_t$ as a finite lag function of $(u_t^2, \sigma_{t+1}^2)$.

It follows from Theorem 4.4 of \citet{Hansen:92} and Assumption \ref{mixing} that
\begin{equation} \label{cgce_1}
\begin{aligned}
T^{-1/2}\sum_{t=1}^{[Tr]}
U_t(\tau)
&  \Rightarrow
\begin{bmatrix}
          \begin{array}{c}
            B_{\psi}(r)\\
            B_v(r)\\
          \end{array}
\end{bmatrix} = BM(0, \Omega_\tau),   \\
T^{-1/2}X_{[Tr]}&  \Rightarrow J_c(r) = \int_0^r e^{(r-\lambda)c}dB_v(\lambda).
\end{aligned}
\end{equation}
The first result in \eqref{cgce_1} corresponds to Assumption A of \citet{xiao:09}.

\subsection{Asymptotic Distribution of the Quantile Regression Estimator}

The following proposition provides the limiting distribution of the predictive quantile regression estimator in (\ref{predictive-quantile-model}). Define $D_T :=\text{diag}(T^{1/2},T)$.
\begin{proposition}
\label{prop_gamma}
Suppose that Assumptions \ref{gamma-const}--\ref{mixing} hold and $x_t$ follows \eqref{x}. Then, we have
\begin{equation}\label{asy_gamma_vec}
D_T(\widehat{\gamma}(\tau)-\gamma(\tau))  \Rightarrow  \frac{1}{f_{u_{\tau}}(0)}\left[\int \overline{J}_c \bar{J}_c'\right]^{-1}\left[\int \overline{J}_c dB_{\psi}\right], 
\end{equation}
where $\overline{J}_c(r):=(1, J_c(r))'$ and $J_c^\mu(r):=J_c(r)-\int_0^1 J_c(s)ds$.
\end{proposition}
It follows from Proposition \ref{prop_gamma} that 
\begin{equation}\label{asy_gamma_1}
T ( \widehat{\gamma}_1(\tau)-\gamma_1)  \Rightarrow  \frac{1}{f_{u_{\tau}}(0)} \frac{\int  J_c^\mu dB_{\psi}}{\int ( J_c^\mu)^2}.
\end{equation}
The asymptotic distribution is nonstandard. When $c=0$, it specializes the result of the quantile cointegrating regression \cite[Theorem 1]{xiao:09} to the case of predictive regression. The extension to $c<0$ was  first derived in \citet{MaynardShimotsuWang2011} under the additional assumption $Q_{u_t}(\tau|\mathcal{F}_{t-1})=Q_{u_t}(\tau)$. \citet{Lee2016} derives the asymptotic distribution when $x_t = (1+c/T)x_{t-1}+v_t$\label{lee_model} with both positive $c$ and $x_t = (1+b/T^\alpha)x_{t-1}+v_t$ with $\alpha \in (0,1)$ (mildly integrated $x_t$ and mildly explosive $x_t$). \citet{FanLee19joe} derive the asymptotic distribution of the predictive quantile regression estimator when $y_t-\gamma_0(\tau)$ contains a conditionally heteroskedastic error of the form $\sigma_t \varepsilon_t$, under the restriction that $\psi_\tau(u_{t\tau})$ is a martingale difference sequence.

As in the case of cointegrating regression, some further insight into the bias can be gained from projecting $B_{\psi}$ onto $B_v$ \citep[][pp.\ 30--31]{Phillips89ET}. Conformable to $(B_\psi,B_v)$, we partition $\Omega_\tau$ into\footnote{We thank Ji Hyung Lee for pointing out a typo in \citet{MaynardShimotsuWang2011} (our earlier working paper), which had $\delta$ in place of $\delta_{\tau}$.}
\[
\Omega_{\tau} = \begin{bmatrix}\begin{array}{cc}
    \omega_{\psi}^2 & \omega_{\psi v} \\
    \omega_{\psi v}   & \omega_{v}^2
  \end{array}
  \end{bmatrix} 
   = \begin{bmatrix}\begin{array}{cc}
    \omega_{\psi}^2 & \delta_\tau \omega_{\psi}\omega_{v} \\
    \delta_\tau \omega_{\psi}\omega_{v}  & \omega_{v}^2
  \end{array}
  \end{bmatrix} ,
\]
where $\delta_\tau:= \omega_{\psi v}/(\omega_{\psi}\omega_{v})$. If $\psi_\tau(u_{t\tau})$ is serially uncorrelated, then $\omega_{\psi}^2$ is simplified to $\tau(1-\tau)$. 

In general, $\psi_\tau(u_{t\tau})$ is serially correlated. For example, when $y_t$  follows a GARCH process, the stochastic process $I(y_t < Q_{y_t}(\tau))$ is serially correlated for $\tau\neq 0.5$ because one large value of $y_{t-1}$ is likely to be followed by another large value of $y_t$. Define $\omega_{\psi.v}^2 := \omega_{\psi}^2 - \omega_{v}^{-2} \omega_{\psi v}^2= \omega_{\psi}^2(1-\delta_\tau^2)$ and $B_{\psi. v} := \omega_{\psi.v}^{-1}(B_{\psi} - \omega_{v}^{-2} \omega_{\psi v} B_v)= \omega_{\psi.v}^{-1}(B_{\psi} - \omega_{v}^{-1} \omega_{\psi}\delta_\tau B_v)$, then $B_{\psi. v}$ is $BM(1)$ and independent of $B_v$.  Using the decomposition $B_\psi = \omega_{v}^{-2} \omega_{\psi v} B_v + \omega_{\psi.v} B_{\psi.v}$, we may express \eqref{asy_gamma_1} as
\begin{equation}\label{asy-breakdown}
T  \widehat{\gamma}_1(\tau)  \Rightarrow \frac{\omega_{\psi}}{f_{u_\tau}(0)}
\left[
\omega_{v}^{-1}\delta_\tau 
\frac{\int J_c^\mu dB_{v}}{\int (J_c^\mu)^2}
+   \sqrt{1-\delta_{\tau}^2}
\frac{\int  J_c^\mu dB_{\psi . v}}{\int (J_c^\mu)^2}
\right].
\end{equation}
The first stochastic integral inside the square brackets is the local-to-unity generalization of the (demeaned) Dickey-Fuller distribution and contributes a downward (upward) second-order bias to the estimate of $\widehat{\gamma}_1(\tau)$ for $\delta_{\tau} >0$ ($\delta_{\tau}<0$). The extent of the bias depends on both $\delta_{\tau}$ and on $c$. The second term in brackets is mixed normal, and normal conditional on $\mathcal{F}_v = \sigma\left( B_v(r), 0\leq r \leq 1\right)$. As in the case of linear predictive regression, there is no endogeneity term because $E\left[\psi_\tau(u_{t\tau})\middle| \mathcal{F}_{t-1} \right]=0$. The distribution of the estimator depends on $\tau$ through both $\delta_{\tau}$ and $B_{\psi.v}$.

We provide HAC standard errors, denoted by $\text{se}(\widehat{\gamma}_1)$, and then derive the asymptotic distribution of the HAC $t$-statistic $(\widehat{\gamma}_1(\tau)-\gamma_1) / \text{se}(\widehat{\gamma}_1)$ when $x_t$ is local-to-unity as in \eqref{x}. Define $\Delta_{fz}(\tau) := E[ f_{u_{t \tau},t-1}(0) z_{t-1} z_{t-1}']$, and define the long-run variance of $w_{t \tau}:=z_{t-1} \psi_\tau(u_{t\tau})$ as $\Sigma(\tau) := \sum_{\ell =-\infty}^{\infty}\Gamma(\ell)$, where $\Gamma(\ell):= E[w_{(t +\ell)\tau}w_{t \tau}' ]$, suppressing the dependence of $\Gamma(\ell )$ on $\tau$. The HAC standard error, $\text{se}(\widehat{\gamma}_1)$, of $\widehat \gamma_1(\tau)$ is defined as the square root of the $(2,2)$th element of $T^{-1} \widehat \Delta_{fz}^{-1}(\tau) \widehat \Sigma(\tau) \widehat \Delta_{fz}^{-1}(\tau)$, where
\[
\widehat \Delta_{fz}(\tau) := \frac{1}{Th} \sum_{t=1}^T  \phi \left(\frac{\widehat u_{t\tau}}{h}\right)z_{t-1} z_{t-1}', \quad \widehat \Sigma(\tau) := \sum_{\ell =-m}^m k \left(\frac{ \ell}{m} \right) \widehat \Gamma(\ell ),
\]
$\widehat{u}_{t\tau} := y_t - \widehat{\gamma}(\tau)'z_{t-1}$, $\phi(\cdot)$  and $k(\cdot)$ are kernels, $h$ is the bandwidth, and $m$ is the lag length. The autocovariance estimate $\widehat \Gamma(\ell )$ is computed as
\[
\widehat \Gamma(\ell ) := \begin{cases}
T^{-1} \sum_{t=1}^{T-\ell } \widehat{w}_{(t +\ell)\tau}\widehat{w}_{t \tau}'  & \ell \geq 0, \\
T^{-1} \sum_{t=-\ell + 1}^{T} \widehat{w}_{(t +\ell)\tau}\widehat{w}_{t \tau}'  & \ell < 0,
\end{cases}
\]
where $\widehat{w}_{t \tau}:=z_{t-1} \psi_\tau(\widehat u_{t\tau})$. When $x_t$ is stationary and follows
\begin{equation}\label{x_stationary}
x_t = \mu_x + \phi x_{t-1} + v_t, 
\end{equation}
with a fixed $\phi \in (-1,1)$, \citet{xiao12hdbk} and \citet{galvao23jasa} show $\sqrt{T}(\widehat \gamma(\tau) -  \gamma(\tau)) \to_d N(0,\Delta_{fz}(\tau)^{-1} \Sigma(\tau) \Delta_{fz}(\tau)^{-1})$ and $(\widehat{\gamma}_1(\tau)-\gamma_1) / \text{se}(\widehat{\gamma}_1) \rightarrow_d N(0,1)$.

We introduce some additional assumptions for analyzing the asymptotics of $\text{se}(\widehat{\gamma}_1)$.
\begin{assumption}\label{kernel}
(a) $\phi(x)$ is bounded, $\phi(x) \leq \mathcal{C}|x|^{-4}$ for some $\mathcal{C} < \infty$, and continously differentiable with the derivative $\phi'(x)$ satisfying $\int \sup_{|x|\leq \epsilon}\left| \phi' \left(s -  x\right)\right| ds < \infty$ for sufficiently small $\epsilon$.
(b) $k(x)$ is continuous, $k(0)=1$, $k(-x)=k(x)$ for all $x$, $\int_{-\infty}^\infty k^2(x) dx < \infty$, $\int_{-\infty}^\infty |x|^{1/2} k(x) dx < \infty$, and $\int_0^{\infty} \overline k(x) dx < \infty$, where $\overline k(x) = \sup_{z>x}|k(x)|$. (c) $1/m + m T^{-1/2}\log T \to 0$. (d) $h + T^{-1/2} \log T/h \to 0$.
\end{assumption}
Many probability density functions, including the standard normal density, satisfy Assumption \ref{kernel}(a). Define $\mathcal{F}^*_{t-1}:= \sigma\{y_{t-j},x_{t-j}, j\geq 1\}$,\label{F_star_defn} and let $f_{u_{t \tau},t-1}^*(x)$ denote the conditional density of $y_t - \gamma_0(\tau)$ conditional on $\mathcal{F}^*_{t-1}$. 
\begin{assumption}\label{f_star}
For each $\tau \in (0,1)$, (a) $\sup_{|x|< \eta}f_{u_{t \tau},t-1}^*(x)< \infty$ with probability one for some $\eta>0$; (b) $\{U_t(\tau), f_{u_{t \tau},t-1}(0), f_{u_{t \tau},t-1}^*(0)\}$ satisfies Assumption \ref{mixing}.
\end{assumption}
The following proposition shows the null limiting distribution of the standard HAC $t$-statistic for testing $H_0: \gamma_1(\tau)=\gamma_1$.\footnote{When $\psi_\tau(\widehat u_{t\tau})$ is serially uncorrelated, this result is originally derived in Proposition 2 of \citet{MaynardShimotsuWang2011}. A similar result is also shown in \citet[][p. 108]{Lee2016}. } The null limiting distribution depends on $c$ and $\delta_{\tau}$.

\begin{proposition}\label{prop_asy_std_t}
Suppose that Assumptions \ref{gamma-const}--\ref{f_star} hold and $x_t$ follows \eqref{x}. Then, we have
\begin{equation} \label{asy_std_t}
t_{\gamma_1}(\tau) := \frac{\widehat{\gamma}_1(\tau) -\gamma_1}{\text{se}(\widehat{\gamma}_1)}  \Rightarrow 
\left[ \omega_v^{-1}  \delta_{\tau} \frac{\int J_c^\mu dB_v}{[\int (J_c^\mu)^2]^{1/2}} + \sqrt{1-\delta_{\tau}^2} Z \right] = Z(c,\delta_{\tau}),
\end{equation}
where $Z \sim N(0,1)$ and is independent of $(B_v, J_c)$. 
\end{proposition}

\section{Inference}
In the stock return predictability example, it is appropriate to model many predictors, such as the dividend price ratio, as near unit root processes as in \eqref{x} with $c<0$. In this case, as shown above, the limiting distribution of $t_{\gamma_1}(\tau)$ is nonstandard and dependent on the nuisance parameters $c$ and $\delta_{\tau}$. When $\delta_\tau \neq 0$, the predictability test with standard normal critical values tends to over-reject the null hypothesis. This is a problem in practice because financial data, such as prices and dividends, do not satisfy strict exogeneity. The over-rejection is especially severe when the residual cross-correlation is large.

For a given value of $c$, the first term in \eqref{asy-breakdown}, which causes the asymptotic bias in $\widehat{\gamma}_1(\tau)$, can be removed using a specialization of the fully modified (FM) approach to the predictive quantile regression framework.  Define 
\begin{align}\label{eq:gamma:+}
\widehat{\gamma}_1(\tau,c)^{+} &:= \widehat{\gamma}_1(\tau) - \frac{\widehat{\omega}_\psi\widehat{\omega}_v^{-1} \widehat{\delta}_{\tau} }{\widehat{f_{u_\tau}(0)} \sum_{t=1}^T (x_{t-1}^\mu)^2 } \left[ \sum_{t=1}^T x_{t-1}^{\mu} (x_t -\phi x_{t-1})  - T\widehat{\lambda}_{vv}\right],
\end{align}
where $\phi = 1+c/T$, and $\widehat{f_{u_\tau}(0)}$, $\widehat{\omega}_\psi$, $\widehat{\omega}_v$, $\widehat{\delta}_{\tau}$, and $\widehat{\lambda}_{vv}$ are consistent estimators of $f_{u_\tau}(0)$, $\omega_\psi$, $\omega_v$, $\delta_{\tau}$, and $\lambda_{vv} := \sum_{k=1}^\infty E v_0 v_k=1/2 (\omega_v^2 - Ev_0^2)$, respectively. Because $T^{-1}\sum_{t=1}^T x_{t-1}^{\mu} (x_t -\phi x_{t-1}) \Rightarrow  \int  J_c^\mu dB_{v}+\lambda_{vv}$, the terms in brackets remove the first term in \eqref{asy-breakdown}. 

\label{report2:DensityEstimaiton}\newcommand{\reportTwoDensityEstimation}{In our simulation and empirical application, we use}\reportTwoDensityEstimation\ a standard kernel density estimator $\widehat{f_{u_\tau}(0)}=1/(Th)\sum_{t=1}^T k(\widehat{u}_{t\tau,t-1}/h)$, where the bandwidth $h$ is chosen by \citet{Silverman:86}'s rule of thumb, and we estimate $\omega_\psi$, $\omega_v$, $\delta_{\tau}$, and $\lambda_{vv}$ nonparametrically. See Section \ref{sec:simulation} for details.

Define the standard error for the fully modified estimator as $\mbox{se}(\widehat{\gamma}_1(\tau,c)^{+}) := \\ (\widehat{\omega}_{\psi.v} / \widehat{f_{u_\tau}(0)}) (\sum_{t=1}^T (x_{t-1}^{\mu})^2 )^{-1/2}$, where $\widehat{\omega}_{\psi.v} = (\widehat\omega_\psi^2(1-\widehat{\delta}_{\tau}^2))^{1/2}$. The following proposition shows that the fully modified estimator has a mixed normal asymptotic distribution and the associated $t$-statistic has a standard normal null asymptotic distribution. 

\begin{proposition}\label{prop_asy_FM}
Suppose that Assumptions \ref{gamma-const}--\ref{mixing} hold and $x_t$ follows \eqref{x}. Then, we have
\begin{align*}
(a) & \quad T ( \widehat{\gamma}_1(\tau,c)^{+} -\gamma_1)
\Rightarrow
\frac{\omega_{\psi.v}}{f_{u_\tau}(0)}
\frac{\int J_c^\mu dB_{\psi.v}}{\int (J_c^\mu)^2} \equiv
MN \left(0, \frac{\omega_{\psi.v}^2}{f_{u_\tau}(0)^2 \int (J_c^\mu)^2} \right) , \\
(b) & \quad t_{\gamma_1}(\tau,c)^+ := \frac{\widehat{\gamma}_1(\tau,c)^{+}-\gamma_1}{\mbox{se}(\widehat{\gamma}_1(\tau,c)^{+})} \to_d N(0,1).
\end{align*}
\end{proposition}
The following corollary provides the asymptotic distribution of the fully modified $t$-statistic for testing $H_0: \gamma_1(\tau)=0$ under a local alternative.
\begin{corollary}\label{corollary_FM}
Suppose that $\gamma_1(\tau) =\gamma_1^*/T$, Assumptions \ref{error-dist1}--\ref{mixing} hold, and $x_t$ follows \eqref{x}. Then, we have
\begin{equation*}
\frac{\widehat{\gamma}_1(\tau,c)^{+}}{\mbox{se}(\widehat{\gamma}_1(\tau,c)^{+})} \to_d \frac{\int J_c^\mu dB_{\psi.v}}{(\int \left( J_c^\mu)^2 \right)^{-1/2}} +\gamma_1^* \frac{f_{u_\tau}(0)}{\omega_{\psi.v}}\left(\int (J_c^\mu)^2 \right)^{1/2}.
\end{equation*}
\end{corollary}

Consider testing $H_0: \gamma_1(\tau)=0$ against $H_A:\gamma_1(\tau) \neq 0$. If the value of $c$ is known, the test that rejects $H_0$ when $|t_{\gamma_1}(\tau,c)^+|>z_{1-\alpha}$ has the asymptotic size $\alpha$. In practice, $c$ is unknown and cannot be consistently estimated. We follow the approach of \citet{cy06} and obtain a conservative testing procedure employing Bonferroni bounds. Inverting the GLS-ADF unit root test of \citet{Elliott&Rothenberg&Stock96} on $x_t$ in the spirit of \citet{stock91} yields a first-stage confidence interval for $c$ with confidence level $\alpha_1$, which we refer to as $\mbox{CI}_c(\alpha_1)$. A Bonferroni test rejects $H_0: \gamma_1(\tau)=0$ in favor of $H_A:\gamma_1(\tau) \neq 0$ if $\max_{c^*\in \mbox{CI}_c(\alpha_1)} |t_{\gamma_1}(\tau,c^*)^+| \geq z_{1-\alpha_2}$, where $\alpha_2$ is chosen so that $\alpha_1+\alpha_2 \leq \alpha$. By the Bonferroni inequality, the asymptotic size of this test is no greater than $\alpha$.\footnote{Equivalently, we can form a confidence interval $\mbox{CI}_{\gamma_1} (\alpha_2,\tau,c)$ of $\gamma_1(\tau)$ with level $1-\alpha_2$ for each value of $c$, define $\mbox{CI}_\gamma(\alpha,\tau):=\cup_{c^*\in \mbox{CI}_c(\alpha_1)}\mbox{CI}_{\gamma_1} (\alpha_2,\tau,c^*)$, and reject $H_0$ when $\mbox{CI}_\gamma(\alpha,\tau)$ does not contain $0$.}

\subsection{Switching-FM Test}\label{sec:switch}

\citet{Phillips14} points out that the \citet{stock91} confidence interval becomes invalid with asymptotic coverage probability zero when $x_t$  is stationary, leading to the invalidity of Bonferroni tests that depend on it. Indeed, when $\delta_{\tau}<0$ and $c$ is large negative, we find that the Bonferroni predictive quantile test becomes extremely conservative against the right-sided alternative and extremely oversized against the left-sided alternative. On the other hand, we found the test to work quite well in the near unit root range. We also noticed similar size problems when $c$ takes very large positive (explosive) values, even though such values are considered implausibly large by most of the predictive regression literature. In view of this, we assume $c \leq 4$ henceforth and rule out implausibly large positive values of $c$.\footnote{\citet{cy06}  assume $c \leq 5$.}

A useful observation is that the persistence ranges for which the Bonferroni test breaks down are also the ranges in which the standard predictive quantile tests work reasonably well. Table \ref{z_table} shows the $5$ and $95$ percentiles of $Z(c,\delta_\tau)$ defined in (\ref{asy_std_t}) for selected values of $(c,\delta_\tau)$. These percentiles are simulated by approximating $B_v(r)$ by $T^{-1/2}\sum_{t=1}^{[Tr]}v_t$ for $v_t\sim i.i.d.\ N(0,1)$ using $1,000,000$ replications with $T=10,000$. As $c \to -\infty$, $Z(c,\delta_\tau)$ converges to $N(0,1)$ \citep[][equation (3) and Section 5]{Phillips14}. When $c \ll 0$, the 5 and 95 percentiles of $Z(c,-1)$ are similar to those of $N(0,1)$, and neither quantile changes very much as $\delta_\tau$ changes. Therefore, when $c \ll 0$, we can use the 5 percentile of a $N(0,1)$ and the 95 percentile of $Z(c^*,-1)$ for $c<c^*$ as conservative critical values for the standard $t$-test without sacrificing much power.

Because the density of $Z(c,\delta_\tau)$ is asymmetric, henceforth we consider separately the right-tailed test of $H_0:\gamma_1(\tau)=0$ against $H_A: \gamma_1(\tau)>0$ with level $(1-\alpha_2/2)$ and the left-tailed test of $H_0:\gamma_1(\tau)=0$ against $H_A: \gamma_1(\tau)<0$ with level $(1-\alpha_2/2)$. To preserve the good properties of the test in the near unit root range while addressing the problems that occur outside it, we propose a switching version of the quantile FM predictive test in the spirit of \citet{ElliotMullerWatson15}. First, consider the right-tailed test of $H_0:\gamma_1(\tau)=0$ against $H_A: \gamma_1(\tau)>0$. Fix a switching threshold $\overline c_{L}<0$. If the first-stage confidence interval $\mbox{CI}_c(\alpha_1)=[\underline c, \overline c]$ lies entirely within the near unit root range, i.e., $\underline{c} > \overline c_L$ we employ only the Bonferroni FM test. When $\mbox{CI}_c(\alpha_1)$ lies entirely outside of the near unit root region, i.e., $\overline{c}< \overline c_L$, we use only the HAC $t$-test with the critical value from $Z(\overline c_L,-1)$. Finally, when $\mbox{CI}_c(\alpha_1)$ lies only partly in the near unit root region, i.e., $\underline{c} \leq \overline c_L \leq  \overline{c}$, we employ both tests and reject the null hypothesis only if both tests reject. In the left-tailed test of $H_0:\gamma_1(\tau)=0$ against $H_A: \gamma_1(\tau)<0$, we fix $\underline c_{L}<0$, which can be different from $\overline c_L$, and proceed similarly to the right-tailed test but use the critical value from a $N(0,1)$ both in the Bonferroni FM test and the $t$-test. This is because the 5 percentile of the $N(0,1)$ serves as the conservative critical value for the $t$-test.

Consider testing $H_0:\gamma_1(\tau)=0$ against $H_A: \gamma_1(\tau)>0$. Define $\overline\varphi_{FM}(\tau,\alpha_1,\alpha_2):=I\{\min_{c^*\in \mbox{CI}_c(\alpha_1)} t_{\gamma_1}(\tau,c^*)^+ \geq z_{1-\alpha_2/2}\}$ and $\overline\varphi_{t}(\tau,\alpha_2,\overline c_L):=I\{ t_{\gamma_1}(\tau) \geq z_{1-\alpha_2/2}(\overline c_L)\}$, where $t_{\gamma_1}(\tau) $ and $t_{\gamma_1}(\tau,c)^+$ are defined as in Propositions \ref{prop_asy_std_t} and \ref{prop_asy_FM} with $\gamma_1=0$, and  $z_{1-\alpha_2/2}(c)$ is the $100(1-\alpha_2/2)$ percentile of $Z(c,-1)$. Define the test statistic
\begin{equation}\label{eq:CI:gamma3}
\begin{aligned}
\overline\varphi(\tau, \alpha_1, \alpha_2, \overline c_L) & :=
\begin{cases}
\overline\varphi_{FM}(\tau,\alpha_1,\alpha_2), & \text{ if } \underline c > \overline c_L, \\ 
\overline\varphi_{FM}(\tau,\alpha_1,\alpha_2)\overline\varphi_{t}(\tau,\alpha_2,\overline c_L),  & \text{ if } \underline c \leq \overline c_L \leq \overline c , \\ 
\overline\varphi_{t}(\tau,\alpha_2,\overline c_L),  & \text{ if } \overline c < \overline c_L. \\ 
\end{cases}
\end{aligned}
\end{equation}
The switching-FM test rejects $H_0:\gamma_1(\tau)=0$ against $H_A: \gamma_1(\tau)>0$ if $\overline\varphi(\tau, \alpha_1, \alpha_2, \overline c_L)=1$. For testing $H_0:\gamma_1(\tau)=0$ against $H_A: \gamma_1(\tau)<0$, defining $\underline\varphi_{FM}(\tau,\alpha_1,\alpha_2):=I\{\max_{c^*\in \mbox{CI}_c(\alpha_1)} t_{\gamma_1}(\tau,c^*)^+ \leq z_{\alpha_2/2}\}$ and $\underline\varphi_{t}(\tau,\alpha_2,\underline c_L):=I\{ t_{\gamma_1}(\tau) \leq z_{\alpha_2/2}\}$ and defining $\underline\varphi(\tau, \alpha_1, \alpha_2, \underline c_L)$ as in \eqref{eq:CI:gamma3} gives the switching-FM test statistic. The two-tailed switching-FM test rejects $H_0:\gamma_1(\tau)=0$ against $H_A: \gamma_1(\tau)\neq 0$ if $\overline\varphi(\tau, \alpha_1, \alpha_2, \overline c_L) + \underline\varphi(\tau, \alpha_1, \alpha_2, \underline c_L)=1$.

The following proposition shows that the asymptotic size of the one-tailed and two-tailed switching-FM test does not exceed $\alpha_1 + \alpha_2/2$ and $\alpha_1 + \alpha_2$, respectively.
\begin{proposition}\label{prop_switch1}
Suppose that Assumptions \ref{error-dist1}--\ref{mixing} hold and $x_t$ follows \eqref{x} or \eqref{x_stationary}. If $\gamma_{1}(\tau)=0$, we have 
\begin{align*}
(a) & \limsup_{T \to \infty} \Pr \left( \overline\varphi(\tau, \alpha_1, \alpha_2, \overline c_L)=1 \right) \leq \alpha_1 + \alpha_2/2, \\
(b) & \limsup_{T \to \infty} \Pr \left( \underline\varphi(\tau, \alpha_1, \alpha_2, \underline c_L)=1 \right) \leq \alpha_1 + \alpha_2/2, \\
(c) & \limsup_{T \to \infty} \Pr \left( \overline\varphi(\tau, \alpha_1, \alpha_2, \overline c_L)+\underline\varphi(\tau, \alpha_1, \alpha_2, \underline c_L)= 1\right) \leq \alpha_1 + \alpha_2. 
\end{align*}
\end{proposition}

In mean predictive regression, \citet{cy06} consider the uniformly most powerful (UMP) test of $H_0:\gamma_1=0$ against $H_1:\gamma_1>0$ when $c$ is known, and their $Q$-test takes a union of the UMP test over $c$ using the Bonferroni method. \citet{ElliotMullerWatson15} establish the optimal test against an alternative model that integrates $c$ and $\gamma_1$ with respect to a user-chosen probability distribution $F$.  Figure 4 of \citet{ElliotMullerWatson15} shows that their test is more powerful than \citet{cy06}'s $Q$-test for many values of $c$, in particular when $c$ is close to 0, but the $Q$-test is more powerful for certain values of $c$. \citet{Hjalmarsson:07} notes that the $Q$-test can be interpreted as a local-to-unity version of the fully modified $t$-test of \citet{Phillips&Hansen90}. Therefore, when $u_t$ follows a Laplace distribution, our fully-modified quantile regression test would be asymptotically optimal when $c$ is known.

\subsection{Adjustments to First-Stage Confidence Levels\label{sec:adjust}}

In mean predictive regression inference, \citet{Cavanagh/Elliott/Stock:95} and \citet{cy06} find that the Bonferroni test, as described above, tends to be excessively conservative. They simulate the asymptotic distribution of their Bonferroni test statistic and adjust the first-stage confidence level $\alpha_1$ to mitigate their test's conservative nature. 

Similar to \citet{cy06}, we simulate the asymptotic distribution of the switching-FM test statistic under the null hypothesis with a large $T$ ($T=5,000$) and adjust the first-stage confidence level $\overline\alpha_1$ for the right-tailed test and $\underline\alpha_1$ for the left-tailed test so that the test is slightly conservative. We fix $\alpha_2$. Let $\tilde \alpha_2 = \alpha_2 - \epsilon$ for a small $\epsilon>0$, and let $Q_c:=[\underline q, \overline q]$ define a region for $c$. Then, we select $\overline\alpha_1$ and $\underline\alpha_1$ over the grid $\{0.01, 0.02, \ldots, 0.98\}$ so that  
\begin{equation} \label{alpha_star}
\begin{aligned}
\limsup_{T \to \infty} \Pr \left(\overline \varphi(\tau, \overline\alpha_1, \tilde\alpha_2 , \overline c_L)= 1 \right) \leq \tilde\alpha_2/2,\\
\limsup_{T \to \infty} \Pr \left(\underline \varphi(\tau, \underline\alpha_1, \tilde\alpha_2 , \underline c_L)= 1 \right) \leq \tilde\alpha_2/2,
\end{aligned}
\end{equation}
holds for all $c \in Q_c$ and with equality for some $c \in Q_c$. The following proposition shows that, when $\underline q$ is chosen sufficiently large negative, this version of the one-tailed switching-FM test has asymptotic size no larger than $\alpha_2/2$ both when $x_t$ is stationary and when $x_t$ is local-to-unity, including the case $c <\underline q$.

\begin{proposition}\label{switch_size}
Suppose that Assumptions \ref{gamma-const}--\ref{mixing} hold, $\underline q$ is sufficiently large negative, and $(\overline\alpha_1,\underline\alpha_1)$ satisfies \eqref{alpha_star}. Suppose that $x_t$ follows \eqref{x} with $c \leq \overline q$  or $x_t$ follows \eqref{x_stationary}. If $\gamma_{1}(\tau)=0$, we have 
\begin{align*}
\limsup_{T \to \infty} \Pr \left( \overline\varphi(\tau, \overline\alpha_1, \tilde \alpha_2, \overline c_L)=1\right) \leq \alpha_2/2, \\
\limsup_{T \to \infty} \Pr \left( \underline\varphi(\tau, \underline\alpha_1, \tilde \alpha_2, \underline c_L)=1\right) \leq \alpha_2/2.
\end{align*}
\end{proposition}

In practice, one needs to choose the value of $\overline c_L$ and $\underline c_L$. Making $\overline c_L$ more negative makes the test less conservative for $c < \overline c_L$ because $z_{1-\alpha_2/2}(\overline c_L)$ becomes smaller. On the other hand, this makes the Bonferroni FM test more conservative for $c > \overline c_L$. 
We set $\alpha_2 = 0.1$, $\epsilon = 0.04$, $Q_c=[-120,4]$ and choose the value of $\overline c_L$ from $\{-120,-110,-100,\ldots,-30\}$ to minimize the average under-rejection\footnote{The under-rejection is calculated as the average of the difference between 0.05 and the rejection frequency. As seen in  Proposition \ref{switch_size}, this difference is always positive.} over $c \in\{-200,-180,-160, \ldots,  0\}$ and $\delta_\tau \in \{-0.797, -0.598, -0.399, -0.199, 0\}$,\footnote{These values of $\delta_\tau$ correspond to the correlation coefficient $\delta \in \{-0.999, -0.75, -0.50, -0.25, 0\}$ between $u_t$ and $v_t$ when $\tau =0.5$.} where $\overline\alpha_1$, which depends on $\overline c_L$, is chosen to satisfy \eqref{alpha_star}. The under-rejection probability is approximated using simulations with $T=5,000$ and $10,000$ replications, where $(u_t,v_t)$ is drawn from a bivariate normal distribution. We thus select $\overline c_L$ to minimize average under-rejection both inside and below $Q_c$ while complying with the requirements of Proposition \ref{switch_size} to avoid any over-rejection. This procedure gives $\overline c_L = -90$. We obtain $\underline c_L = -100$ by a similar procedure.

Table \ref{table:alpha:1:A} shows the resulting adjusted significance levels $\overline{\alpha}_1$  and $\underline{\alpha}_1$ used for the confidence intervals on $c$ for the right and left-tailed predictive tests, respectively. Table \ref{table:alpha:1:A} also shows the values used for $\delta$ and the corresponding values of $\delta_{\tau}$. $\delta$ and $\delta_{\tau}$ are closely related, although $\delta_{\tau}$ is generally smaller in magnitude than $\delta$. We use the value of $\delta_{\tau}$ when employing these lookup tables.  Both $\underline{\alpha}_1$ and $\overline{\alpha}_1$ are smaller when  $|\delta_{\tau}|$ is larger and the persistent regressor problem is worse. Even then, however, they are well above five percent, suggesting that the adjustment can help to reduce the conservativeness of the Bonferroni test procedure. 

\section{Power of the Switching-FM Test under Local Quantile-Varying Alternatives}\label{sec:local_power}

In this section, we analyze the asymptotic power of the switching-FM test under local quantile-varying alternatives with local-to-unity $x_t$. In order to facilitate asymptotic analysis, we model $y_t$ as a random-coefficient process similar to \citet{KoenkerXiao2006}:
\begin{equation} \label{model_local}
y_t = \gamma_0 + \gamma_1 x_{t-1} + e_t  + T^{\kappa-1}b(e_t) |x_{t-1}+\zeta|, \quad \kappa \in (0,1/2), 
\end{equation}
where $e_t$ is independent of $\mathcal{F}_{t-1}$, $b(\cdot)$ is weakly monotone increasing with $\sup_{x,y} [b(x)-b(y)]/(x-y) = M_b < \infty$, and $\zeta>0$ is non-random. When $\zeta$ is sufficiently large, $x_{t-1}+\zeta>0$ holds for all $t$ in the observed data, and the absolute value hardly matters in practice.

The predictive quantile for $y_t$ 
\begin{align*}
  Q_{y_t}(\tau | \mathcal{F}_{t-1}) &= \gamma_0(\tau) + \gamma_1 x_{t-1} + T^{\kappa-1}b (Q_{e_t}(\tau)) |x_{t-1}+\zeta|,
\end{align*}
where $\gamma_0(\tau) = \gamma_0 + Q_{e_t}(\tau)$, has a local quantile-varying component. Define $e_{t\tau}:=e_t - Q_{e_t}(\tau)$. Under (\ref{model_local}), we have
\begin{equation}\label{uttau_new}
u_{t\tau} = y_t - Q_{y_t}(\tau | \mathcal{F}_{t-1}) = e_{t\tau} + T^{\kappa-1}[b (e_t)- b (Q_{e_t}(\tau))] |x_{t-1}+\zeta| .
\end{equation}
Because $\kappa<1/2$, $e_{t\tau}$ dominates the right hand side, and $u_{t\tau}$ behaves like an $I(0)$ process. We show that, with some restrictions on $\kappa>0$, the switching-FM test rejects $H_0:\gamma_1(\tau)=0$ with probability approaching one.

We collect assumptions. Assumption \ref{local_e-density} corresponds to Assumption 2.1(ii) in \citet{Lee2016}.
\begin{assumption}
\label{local_e-density}
$E|e_t|<\infty$ holds. Further, for each $\tau \in (0,1)$, the density of $e_{t\tau}$, $f_{e_{\tau}}(\cdot)$ is bounded and continuous around $0$ and $f_{e_{\tau}}(0)>0$.
\end{assumption}

\begin{assumption}
\label{local_mixing}
Assumption \ref{mixing} holds when $(u_{t\tau},f_{u_{t\tau},t-1}(0))$ is replaced with $(e_{t\tau},f_{e_{\tau}}(0))$.
\end{assumption}

The following proposition shows the asymptotic distribution of the QR estimator under local alternatives (\ref{model_local}). Define $D_T^\kappa := \text{diag}(T^{1/2-\kappa},T^{1-\kappa})$ and $\gamma(\tau):=(\gamma_0(\tau), \gamma_1)'$.

\begin{proposition}\label{prop_local}
Suppose $y_t$ follows (\ref{model_local}) and Assumptions \ref{local_e-density}--\ref{local_mixing} hold. Then,
\begin{align*}
D_T^{\kappa}(\widehat{\gamma}(\tau)-\gamma(\tau)) & \Rightarrow 
b(Q_{e_t}(\tau))
\left[
\int \overline J_c \overline J_c' 
\right]^{-1}
\left[
\int \overline J_c |J_c| 
\right], \\
T^{1-\kappa}(\widehat{\gamma}_1(\tau)-\gamma_1) & \Rightarrow 
b(Q_{e_t}(\tau))
\frac{\int |J_c| J_c^{\mu}}{(\int J_c^{\mu})^2}.
\end{align*}
\end{proposition}
The asymptotic distribution is a functional of $b(Q_{e_t}(\tau))$ and $J_c(r)$. Under the local alternative with $\gamma_1=0$, $\widehat{\gamma}_1(\tau)$ converges to 0 at a slower rate $T^{\kappa-1}$ than $T^{-1}$. 

The following proposition shows the consistency of the switching-FM test under local alternatives (\ref{model_local}). The additional assumption $T^{\kappa-1/2}/h \to 0$ is necessary to control the convergence rate of the standard error. When one uses the optimal bandwidth $h = C T^{-1/5}$, the restriction on $\kappa$ becomes $\kappa \in (0,3/10)$, which is fairly weak.

\begin{proposition}\label{switch_power}
Suppose $y_t$ follows (\ref{model_local}) and $b (Q_{e_t}(\tau)) \neq 0$. If $T^{\kappa-1/2}/h \to 0$, the switching-FM test of $H_0:\gamma_1(\tau)=0$ against $H_A:\gamma_1(\tau) \neq 0$ rejects $H_0$ with probability approaching one.
\end{proposition}

\section{Simulation Study}\label{sec:simulation}

Our Monte Carlo study has three primary objectives. First, we examine the extent of the size distortion of a standard predictive quantile $t$-test. Second, we study the size performance of the proposed switching-FM test across different levels of persistence and serial correlation, including for cases of volatility persistence, leverage, and fat-tails. Finally, we compare the small sample size and power of the switching-FM test to the $t^w$ test of \citet{cai23joe}.\footnote{Comparisons of the $t^w$ test to \citet{Lee2016}'s IVXQR test can be found in \cite{cai23joe}, where $t^w$ is generally found to have better size and higher power than IVXQR.}
We evaluate test power under both the traditional linear regression alternative and under two random coefficient models: one from Section \ref{sec:local_power}, in which only the shoulder and tail quantiles are predictable, and a second specification from \citet{cai23joe}, in which predictability is stronger in the upper quantiles than in the lower quantiles.

When implementing the switching-FM test, we estimate $\omega_\psi$, $\omega_v$, $\delta_{\tau}$, and $\lambda_{vv}$ by kernel-based heteroskedasticity and autocorrelation consistent (HAC) estimators. We use a Gaussian kernel for $\phi(\cdot)$ and the Bartlett kernel for $k(\cdot)$. Because $\psi_\tau(u_{t\tau})$ has persistent serial correlation when $y_t$ follows a GARCH process, we apply a heterogeneous VAR (HVAR)-prewhitening \citep{Corsi09jfemets} to $U_t=(\psi_\tau(u_{t\tau}),v_t)'$. Specifically, we first fit a HVAR model $U_t = \beta_0 + \Phi^{(m)} U_{t-1}+ \Phi^{(q)} U_{t-1}^{(q)} + \Phi^{(y)} U_{t-1}^{(y)} + \ve_t$ to $U_t$, where $U_{t-1}^{(q)} = (1/3)\sum_{j=1}^3 U_{t-j}$ and $U_{t-1}^{(y)} = (1/12)\sum_{j=1}^{12} U_{t-j}$.  Because $v_t$ has a weak serial correlation in predictive regression models, and $\psi_\tau(u_{t\tau})$ is uncorrelated with lagged $v_t$'s by its definition, we impose zero restrictions on the $(1,2)$th element of $\Phi^{(m)}$, $\Phi^{(q)}$, and $\Phi^{(y)}$ and the second row of $\Phi^{(q)}$ and $\Phi^{(y)}$. After estimating the long-run variance matrix $\Omega_\ve$ of $\ve_t$ by a HAC estimator with the \citet{Andrews91} bandwidth choice, we obtain $\widehat \Omega_\tau$ by recoloring $\widehat{\Omega}_\ve$ with the estimate of $\Phi^{(m)}$, $\Phi^{(q)}$, and $\Phi^{(y)}$. $\lambda_{vv}$ is estimated similarly. The bandwidth $h$ in $\widehat \Delta_{fz}$ is chosen by \mcitet{Silverman:86}'s rule of thumb.

We simulate the predictor from an autoregressive model of order one, as in \eqref{x}, employing $x_0=0$ as the starting value. We consider the odd decile values of the quantile level $\tau = 0.1, 0.3,\ldots, 0.9$ for sample sizes of $800$ and $1600$. All simulations are based on 10,000 Monte Carlo replications. Table \ref{table:size} shows null rejection rates of nominal five percent tests of $H_0:\gamma_1(\tau)\leq 0$ against $H_A:\gamma_1(\tau)>0$ when $y_t$ is generated by $y_t=e_t$ and $(v_{t},e_{t})$ follows an i.i.d.\ bivariate normal distribution with means equal to zero, unit variances, and correlation $\delta$. We select $\delta<0$ because empirical estimates of $\delta$ using valuation predictors, such as the earning price ratio or dividend price ratio, are generally negative. Similarly, we conduct one-sided versions of all the predictive tests since $\gamma_1>0$ is the relevant alternative hypothesis in empirical work using valuation predictors. We compare three tests: the conventional quantile regression $t$-test without HAC,\footnote{We use a non-HAC version of the standard $t$-statistic in Table \ref{table:size} because $e_t$'s in \eqref{model_local} are serially independent.} the switching FM-test, and \mcitet{cai23joe}'s $t^w$ test. The quantile levels ($\tau$) vary across the table columns. We vary the local-to-unity parameter ($c$) across rows, including both the standard near unit range $-25<c\leq 0$ and $c=-200$, which entails more stationary behavior. The endogeneity is strong in the top two panels ($\delta=-0.95$) but moderate in the bottom two panels ($\delta=-0.5$). We increase the sample size from $T=800$ in the left panels to $T=1600$ in the right panels.

The results in Table \ref{table:size} indicate that the size problem in conventional predictive quantile regressions can be non-trivial, with rejection rates as high as $0.289$ even for a sample size of $1600$.  This confirms both our original finding in \citet{MaynardShimotsuWang2011} and subsequent results in \citet{Lee2016} and \cite{cai23joe}. The degree of size distortion depends heavily on both the local-to-unity parameter $c$ and the residual correlation $\delta$ (which impacts $\delta_{\tau }$). The table also confirms that the size distortion gradually dissipates as $c$ becomes more negative. This finding supports the use of the switching-FM test. It is also in line with \citet{Lee2016}'s theoretical finding that the distribution of quantile regression estimator becomes standard normal in the mildly integrated case.


The rejection rates of both the switching-FM test and the $t^w$ test are much closer to their target five percent level. The accurate rejection rates for $c=-200$ also address the critique of \citet{Phillips14}. In both tests, we do observe a slight over-sizing in the outermost deciles when $T=800$. However, this improves when the sample size increases to $T=1600$. Overall, the size performance of the two tests is similar. The only discernible differences are that the $t^w$ test performs somewhat better in the outermost quantiles when $c\leq-5$, while the switching-FM test is slightly less conservative when $c=-25$.


In Table \ref{tab:GJR}, we examine the size of the switching-FM test when $y_t=e_t$ and $e_t$ has conditional heteroskedasticity, leverage, and fat-tails. We do not include the $t^w$ test in this table since it is not designed for the case when $\psi_\tau(u_{t\tau})$ has autocorrelation. When $e_t$ follows a GARCH process and correlates with $v_t$, the outer population quantiles of $y_t$ can be predicted using $x_{t-1}$. This places us outside of the null hypothesis. Therefore, we consider the following two models; (A) $e_t$ follows a GJR-GARCH(1,1)-t($\nu$) process that is independent of $v_t$, and (B) $e_t$ has fat tails, is correlated with $v_t$, but has no conditional heteroskedasticity.  In model A, we use the estimated parameters of a GJR-GARCH(1,1)-t($\nu$) model fitted to the monthly stock returns from our empirical application and generate $e_t$ as $e_t = \sigma_t\ve_{t}$  where
\begin{equation}
\sigma_t^2 = 0.0001+ 0.0558e_{t-1}^2 +0.1382 I(e_{t-1}<0)e_{t-1}^2+ 0.8226 \sigma_{t-1}^2, \label{eqn:GJR}
\end{equation}
and $(\ve_{t},v_t)$ is drawn from mutually independent and i.i.d.\ $t$-distributions with $\nu$ degrees of freedom. In model B, $(v_t,e_t)$ are jointly drawn from an i.i.d.\ multivariate $t$-distribution with degrees of freedom $\nu$ and correlation $\delta=-0.95$. All innovations are rescaled to have unit variances.

Table \ref{tab:GJR} has two panels. Panels A and B report the results with models A and B, respectively.  Using our empirical returns, we estimate $8.64$ degrees of freedom in our GJR-GARCH-$t$ residual. Rounding this value down, we set $\nu=8$ in Panel A. In Panel B,  we use $\nu=3$, which rounds down our estimate of $3.67$ when fitting a $t$-distribution without GJR-GARCH to our empirical returns. In both panels of Table \ref{tab:GJR}, the size of the switching-FM test is close to the i.i.d.\ case when $0.3 \leq \tau \leq 0.7$. In Panel A, there is some finite sample size distortion in the outer deciles, but the over-rejection is not severe even for $T=800$ and improves for $T=1600$. Therefore, the switching-FM test performs satisfactorily for a realistically calibrated GJR-GARCH-$t$ model that allows for volatility, leverage, and fat-tails. The finite sample size distortion in the tail is a bit stronger in Panel B, which is not surprising given the very fat tails of Panel B. Nonetheless, this size distortion also improves considerably when increasing the sample size from $800$ to $1600$.



In Table \ref{table:power:n400:high:delta}, we examine the finite sample power of the switching-FM test and compare its performance to that of the $t^w$ test under the traditional linear alternative $y_t =  \gamma_1 x_{t-1} + e_t$ with Gaussian errors as in Table \ref{table:size}. We consider the local alternative $\gamma_1 =\gamma_1^*/T$, analyzed in Corollary \ref{corollary_FM}, for $\gamma_1^* \in \{5, 10, 15, 20, 25\}$ using $T=800$ (left-side panels) and $T=1600$ (right-side panels).  The value of $\tau$ is varied across the rows. To save space, we show only the median and outer deciles. The unshaded columns (2 and 8) provide null rejection rates ($\gamma_1^*=0$), while the shaded columns (Columns 3--7 and 9--13) provide rejection rates under the alternative hypothesis that $\gamma_1^*>0$. The top, middle, and bottom panels show results for three different values of the local-to-unity parameter: $c=-5$ (Panels A--B), $c=-10$ (Panels C--D), and $c = -25$ (Panels E--F). We show only results for $\delta=-0.95$ since overall power comparisons for $\delta=-0.50$ are similar. As can be seen from Table \ref{table:power:n400:high:delta}, both tests perform well. Their power increases reliably as $\gamma^*$ increases. Their power remains approximately constant as $T$ increases and $\gamma^*/T$ shrinks, demonstrating power against the local alternative. While the $t^w$ test has a modest size advantage in the tails when $c=-5$, overall, the FM-switching test has higher power across all six panels, with substantial differences in some cases.

In Table \ref{table:power:tail:high:delta}, we next compare power under an alternative hypothesis for which there is predictability in the tails and shoulders but no median predictability. We again test $H_0:\gamma(\tau)\leq 0$ against $H_A:\gamma(\tau)>0$ using both the switching-FM and $t^w$ tests with $y_t$ generated according to \eqref{model_local}. To exclude median predictability we set $\gamma_1=0$, but allow for predictability at other quantiles by setting $\kappa=0.25$, $b(e_t) = be_t$, and $\zeta = 25$. We vary the value of $b$ across columns. Under this specification, $\gamma(\tau)=0$ holds for either $\tau=0.5$ or $b=0$, whereas $\gamma(\tau)>0$ holds only when both $\tau>0.5$ and $b>0$. Thus, the null rejection rates are shown in the unshaded regions corresponding to the union of $b=0$ and $\tau=0.5$ of each sub-panel. Finite sample power is shown in the shaded regions for which $\tau>0.5$ and $b>0$. Since we employ one-sided tests, we do not show results for $\tau<0.5$, for which $\gamma(\tau)<0$.

Not surprisingly, the size results in column 2 for $b=0$ are similar to those in the previous tables. The null rejection rates in the rows associated with $\tau=0.5$ when $b>0$ are new to this table. In this conditionally heteroskedastic case, increases in $x_{t-1}$ are associated with increased residual variance. An example of $x_{t-1}$ would be a volatility predictor such as the realized variance. Encouragingly, the size of both tests remains stable across the columns. In this model, we move away from the null hypothesis by increasing either $b$ or $\tau$. Reassuringly, the power of both the tests improves with an increase in either parameter.  Comparing the left and right panels, we also see that the power of both tests increases when the sample size increases, which corroborates the consistency of the switching-FM test shown in Proposition \ref{switch_power}. Lastly, we note that the power of the switching-FM test again exceeds that of the $t^w$ test.

Finally, in Table \ref{table:power:randcoef:cai}, we simulate $y_t$ from the same random coefficient model employed by \cite{cai23joe}. Specifically, we generate $x_t$ from the autoregressive model in \eqref{x} and generate $y_t$ from the random coefficient model:
  \begin{equation}\label{DGP:cai}
  y_t =3(1 + T^{-1}\gamma^*
  x_{t-1} )+  (1 + T^{-1}\gamma^* x_{t-1} ) e_t,
  \end{equation}
where the $(v_{t},e_{t})$ are i.i.d.\ innovations from a bivariate normal distribution with means equal to zero, unit variances, and correlation $\delta$.  
The model implies a value of $\gamma_1(\tau)= T^{-1}\gamma^*\left(Q_{e_t}(\tau)+3\right)$ in \eqref{predictive-quantile-model} (see \cite{cai23joe}).
Under this alternative, for $\gamma^*>0$, an increase in $x_{t-1}$ can impact all the quantiles of $y_t$, but it has a larger positive impact on the upper quantiles of $y_t$ than it does on the lower quantiles. Indeed, the results in Table \ref{table:power:randcoef:cai} confirm that the power of both tests is strongest at the ninth decile and weakest at the first decile. However, in all cases, the power of the tests increases with both $\gamma^*$ and with sample size. The switching-FM test is generally more powerful than the $t^w$ test, while previous results in \cite{cai23joe} show that the $t^w$ test is more powerful than the IVXQR test. 
 
To summarize our Monte Carlo results, standard predictive quantile regression $t$-tests can suffer severe over-rejection when the predictor is both persistent and endogenous. The switching-FM predictive quantile test proposed here successfully solves this problem. Furthermore, it also provides good size when the predictor is far less persistent, thereby addressing the critique of \citet{Phillips14}. Overall, the size results for the switching-FM test are quite good, matching those of the $t^w$ test at all but the outermost quantiles. Our simulations further confirm the robustness of the switching-FM test to conditional heteroskedasticity and fat-tails. Finally, the switching-FM test is found to have better power than the  $t^w$ test in both the linear alternative and in two different random coefficient models. Since \cite{cai23joe} have previously shown the $t^w$ test to have higher power than the IVXQR test, we find these results promising. On the other hand, both the $t^w$ and IVXQR tests are readily extended to multiple predictors, whereas the switching-FM test is applicable for only a single predictor. 

\section{Empirical Study}\label{sec:empirical}
We apply the switching-FM predictive method developed above to test for predictability at different points in the stock return distribution. We employ monthly data from \citet{Goyal:Welch:2008}  updated to 2015 and focus on the valuation based predictors for which the predictive regression problem is most pertinent.\footnote{While some of the other predictors employed by \citet{Goyal:Welch:2008} are also persistent, estimates of $\delta$ generally indicated little endogeneity, implying little distortion from the standard quantile tests. }
 Namely, we separately employ lagged values of the dividend price ratio ($dp_t$), the earnings price ratio ($ep_t$), and the book-to-market ratio ($bm_t$) as univariate predictors. We test their ability to predict the center, shoulders, and tails of the return distribution, using value weighted monthly excess returns including dividends on the S\&P 500 from the Center for Research on Security Prices (CRSP). The data runs from January 1926 until December 2015.\footnote{We thank Amit Goyal for use of his publicly available data and for answers to several queries, including confirmation that the CRSP monthly return series with dividends is unavailable prior to 1926.}

Compared to the vast literature on mean prediction, there have been relatively few applications of predictive quantile regression. \citet{Cenesizoglu:Timmermann:08}  employed predictive quantile methods with data on 16 predictors from \citet{Goyal:Welch:2008} ending in 2005. However, they used standard predictive quantile tests without addressing their size distortion. In \citet{MaynardShimotsuWang2011}, we revisited these results using an earlier version of our Bonferroni test. Subsequent studies have employed the IVXQR \citep{Lee2016,FanLee19joe}, the maximized Monte Carlo test \citep{GungorLugo2019}, and a test based on an auxiliary regressor \citep{cai23joe}. 

In Table \ref{table:indicators}, we first provide preliminary indications of the persistence and endogeneity of the three valuation predictors. In Row 2, we provide the $t$-statistic for the GLS-ADF unit root test of \citet{Elliott&Rothenberg&Stock96} using BIC to select the lag length. At the 5\%  significance level, we can reject a unit root in the earnings price ratio but fail to reject for the other two predictors.  Under the assumption of a local-to-unity model for the predictors, we next construct a 95\% confidence interval ($c_L^{0.95},c_U^{0.95}$) for the local-to-unity parameter by inverting the GLS-ADF $t$-test following the approach of \citet{stock91}.  The lower bounds $c_L^{0.95}$ (Row 4) range between $-10$ and $-20$, confirming that all three variables have near unit roots. The upper bound $c_U^{0.95}$ (Row 5) for the earnings price ratio is only very slightly below zero, whereas it just slightly exceeds zero for the other two predictors. The values of $\phi$ implied by the lower and upper bounds on the confidence region for $c$ are provided in Rows 6--7. Overall, these results confirm that all three predictors are well-modeled as near-unit root processes.

The final row of Table \ref{table:indicators} provides the sample correlation coefficient between $\widehat{\varepsilon}_{t}$ and $\widehat{u}_{t}$, where $\widehat{\varepsilon}_{t}$ is the residual from the ADF regression on $x_t$, using BIC to select lag-lengths, and $\widehat{u}_{t}$ is the residual from regressing the stock return on the predictor. The estimated residual correlations are large and negative for all three predictors. This combination of persistence and endogeneity suggests nontrivial size distortion in standard quantile predictive regression.

Table \ref{table:fm-switch} provides the empirical results from the switching-FM predictive quantile regression tests. To assess the ability of the valuation predictors to predict at different points in the return distribution, including the center, shoulders, and tails, we conduct the test at each of the nine deciles shown in the top row. The three panels display the test results for the log dividend price ratio, earnings price ratio, and book-to-market ratio.

The first two rows of each panel in Table \ref{table:fm-switch} show the standard quantile regression slope coefficient and its non-HAC $t$-statistic, without bias or size correction. As discussed earlier, this standard, non-HAC $t$-statistic suffers from two problems: serial correlation in $\psi_\tau(u_{t\tau})$ and a nonstandard asymptotic distribution. The HAC version of the standard $t$-statistic in the third row addresses the first problem, but its distribution is still nonstandard. Row 4 of each panel reports estimates of the long-run residual correlation, $\delta_{\tau}=\omega_{\psi v}/(\omega_{\psi}\omega_{v})$, that are large enough to induce substantial size distortion in the standard HAC quantile $t$-test. Rows 5--6 of each panel of Table \ref{table:fm-switch} report the first-stage confidence interval on $c$ computed using the adjusted significance levels $\underline{\alpha}_1$ and $\overline{\alpha}_1$ in Table \ref{table:alpha:1:A} corresponding to $\widehat{\delta}_\tau$.\footnote{Since $\underline{\alpha}_1$ and $\overline{\alpha}_1$ are adjusted upwards (see Section \ref{sec:adjust}), they are tighter than the 95\% confidence intervals $(c_L^{0.95},c_U^{0.95})$ from Table \ref{table:indicators}. They also depend on the quantile level $\tau$.} For this dataset, all of our first-stage confidence intervals on $c$ lie within the near unit root range.

The final two rows report the Bonferroni confidence interval for the slope coefficient $\gamma_1(\tau)$ obtained by inverting the switching-FM test. When the lower bound is positive, i.e., $\underline{\gamma}_1(\tau)>0$, the switching-FM test rejects the null hypothesis of no predictability against $H_A: \gamma_1(\tau)>0$ at the 5\% significance level. Similarly, an upper bound below zero, $\overline{\gamma}_1(\tau)<0$, implies a left-sided rejection at the 5\% significance level. The cases for which either occurs are marked in bold.

The test results for the dividend-price ratio show an interesting pattern across quantiles. \citet{Goyal:Welch:2008} have argued that in-sample mean-predictability from the dividend price ratio is heavily reliant on observations from the oil crisis period of the early 1970s and disappears in later samples. Similarly, our switching-FM predictive quantile test is unable to reject the null hypothesis of no median predictability using the dividend price ratio (top panel). Given the robustness of quantile regression, this lends additional support to \citet{Goyal:Welch:2008}'s earlier results. 

It would nonetheless be premature to conclude that the dividend-price ratio lacks useful predictive content for returns. In fact, the switching-FM test shows it to be predictive for the upper shoulder of the return distribution. An increase in the dividend price ratio corresponds to a larger right shoulder and tail for the return distribution. Since the left shoulder is insignificant, the overall predictive pattern differs from that of a pure increase in volatility but is in line with the informal notion that a low market valuation (low price relative to dividend) may set the stage for future market rallies.  For example, valuations may be low during a ``bear'' market, and the right shoulder of the distribution may reflect the possibility of a strong market recovery.   

Our results also underline the importance of robustifying inference to both the serial correlation of the quantile-regression-induced residuals and the size distortion resulting from the persistence and endogeneity of the predictors.  The standard non-robust $t$-statistic in the second row naively indicates predictability in the tails and/or shoulders using all three predictors.  However, after applying HAC standard errors in row 3, the significance is lost in several cases, including the left-tail for the dividend price ratio and the right-tail for the book-to-market ratio. Finally, after using the switching-FM test, only the right shoulder of the dividend price ratios remains significant. Consequently, using a standard predictive quantile test, even with HAC standard errors, can greatly exaggerate the evidence of quantile predictability. Thus, the use of both robust standard errors and proper inference procedures is essential in empirical applications involving predictive quantile regression.

\section{Conclusion}
This paper develops inference in predictive quantile regressions with a nearly nonstationary regressor.  We derive the limiting distributions of the quantile regression coefficient and its corresponding HAC $t$-statistic under a local-to-unity specification for the predictor. The asymptotic analysis suggests size distortion using standard tests of quantile predictability. Our simulations indicate that when the predictor is both persistent and endogenous, the size distortion can be nearly as serious as that of predictive (mean) regression. \citet{MaynardShimotsuWang2011} was the first to identify these issues. \citet{Lee2016} has subsequently generalized these results to the mildly integrated and explosive cases.

One of the challenges to correcting inference in predictive regression is the inability to consistently estimate the local-to-unity parameter. A popular solution in the mean regression case has been the use of a Bonferroni bound in conjunction with a first-stage bound on this parameter. While this works well in practice when the predictor has a near unit root, \citet{Phillips14} has recently shown that it becomes invalid when the predictor is stationary. We, therefore, propose a switching-fully modified (FM) predictive quantile regression test that uses a Bonferroni bound over an FM style bias-corrected quantile regression estimator when the predictor has a near unit root and switches to a standard (quantile) test with slightly conservative critical values when the predictor is stationary.

Our simulations indicate that this method works well in practice over a wide range of values for the largest autoregressive root, including values slightly below one (near unit root) and far below one (stationary). They also demonstrate that while both our test and the $t^w$ predictive quantile test of \citet{cai23joe} perform well in the finite sample, the switching-FM test can compare favorably in terms of power while maintaining roughly similar finite sample size. Previous results by \citet{cai23joe} show their $t^w$ test to have better size and power than IVXQR. On the other hand, both the $t^w$ and IVXQR tests have the practical advantage of easily generalizing to multivariate predictors and the $t^w$ test has a modest size advantage at tail quantiles. 

We test the predictability of three heavily employed valuation predictors, the dividend price ratio, earnings price ratio, and book-to-market ratio, for monthly returns on the S\&P 500. Our data spans 1927-2015, which includes the financial crisis and its aftermath, a particularly interesting period when considering the quantiles of the return distribution. We find no significant evidence of predictability in the center of the distribution using even the standard predictive quantile test, and, for two of our three predictors, findings of predictability at other quantiles from the standard test are overturned by our switching-FM test. Nonetheless, we find significant evidence that the dividend price ratio is predictive for the right shoulder of the return distribution.

Extensions of the switching-FM approach to the case of multiple predictors and joint tests could provide an interesting but non-trivial direction for future research. The design of a suitable switching rule with multiple predictors could prove challenging when some regressors are persistent and others are not. The practical implementation of the Bonferroni bound, especially the adjustments to the first-stage confidence interval needed to keep it from being overly conservative, could also be considerately complicated in even a bivariate case. In view of our empirical finding of in-sample predictability at certain non-central quantiles, the development of tests for the out-of-sample predictive power of quantile prediction with persistent and endogenous regressors could provide a second promising direction for future research.

\newpage
%


  

\newpage
\begin{appendices}

\section{Proofs}

\subsection{Proof of Proposition \ref{prop_gamma}}

The proof follows the argument in the proof of Theorem 1 of \citet{xiao:09}. Recall $u_{t\tau} = y_t - \gamma(\tau)'z_{t-1}$. Observe that minimization \eqref{qr_estimator} is equivalent to $\min_v Z_T(v)$, where $Z_T(v) := \sum_{t=1}^T \left[ \rho_\tau(u_{t\tau}- (D_T^{-1}v)'z_{t-1}) - \rho_\tau(u_{t\tau})\right]$. If $\widehat v$ is a minimizer of $Z_T(v)$, then we have $\widehat v = D_T(\widehat{\gamma}(\tau)-\gamma(\tau))$. The objective function $Z_T(v)$ is convex. Therefore, as shown in \citet{Knight1989} and \citet{pollard91et}, if the finite-dimensional distributions of $Z_T(\cdot)$ converge weakly to those of $Z(\cdot)$ and $Z(\cdot)$ has a unique minimum, then $\widehat v$ converges in distribution to the minimizer of $Z(\cdot)$.

As shown in \citep[][p.\ 256]{xiao:09}, for any $u\neq0$, $\rho_\tau(u-v)-\rho_\tau(u) = -v\psi_\tau(u)+(u-v)\{I(0>u>v)-I(0<u<v)\}$. It follows that \citep[][p.\ 257]{xiao:09}
\begin{equation} \label{ZT}
\begin{aligned}
Z_T(v) &= -\sum_{t=1}^T (D_T^{-1}v)'z_{t-1}\psi_\tau(u_{t\tau}) \\
& \quad - \sum_{t=1}^T (u_{t\tau}-(D_T^{-1}v)'z_{t-1}) I(0<u_{t\tau}<(D_T^{-1}v)'z_{t-1})\\
& \quad +\sum_{t=1}^T (u_{t\tau}-(D_T^{-1}v)'z_{t-1}) I(0>u_{t\tau}>(D_T^{-1}v)'z_{t-1}) .
\end{aligned}
\end{equation}
We proceed to derive the asymptotics of each term on the right hand side. For the first term on the right hand side of \eqref{ZT}, it follows Assumption \ref{mixing} and Theorem 4.4 of \citet{Hansen:92} that
\begin{align*}
D_T^{-1}\sum_{t=1}^T z_{t-1}\psi_\tau(u_{t\tau}) &= 
	\begin{bmatrix}
    	T^{-1/2}\sum_{t=1}^T\psi_\tau(u_{t\tau})\\
        T^{-1}\sum_{t=1}^T x_{t-1}\psi_\tau(u_{t\tau})
        \end{bmatrix}
        \Rightarrow 
            \int_0^1 \overline{J}_c d B_{\psi}.
\end{align*}
Therefore, the first term on the right hand side of \eqref{ZT} converges to $-v'\int_0^1 \overline{J}_c d B_{\psi}$ in distribution. Similar to \citet[][p.\ 257]{xiao:09}, write the second term on the right hand side of \eqref{ZT} as
\begin{align*}
W_T(v) & :=\sum_{t=1}^T \xi_t(v), \text{ where }\ \xi_t(v)  =(v'D_T^{-1}z_{t-1}-u_{t\tau})
I(0<u_{t\tau}<v'D_T^{-1}z_{t-1})
\end{align*}
We consider the truncation of $v'D_T^{-1}z_{t-1}$ at some finite number $m>0$. Define $W_{Tm}(v):=\sum_{t=1}^T \xi_{tm}(v)$, where $\xi_{tm}(v):=\xi_{t}(v)I(v'D_T^{-1}z_{t-1}\leq m)$. Denote $\overline{\xi}_{tm}(v):=E[\xi_{tm}(v)|\mathcal{F}_{t-1}]$ and $\overline{W}_{Tm}(v):=\sum_{t=1}^T \overline{\xi}_{tm}(v)$. In view of $z_{t-1}\in \mathcal{F}_{t-1}$, Assumption \ref{error-dist1}, and $\max_{1\leq t\leq T}(v'D_T^{-1}z_{t-1})I(0\leq v'D_T^{-1}z_{t-1}\leq m) =o_p(1)$, we have
\begin{align*}
\overline{W}_{Tm}(v)
&= \sum_{t=1}^T E \left[(v'D_T^{-1}z_{t-1}-u_{t\tau}) I(0<u_{t\tau}<v'D_T^{-1}z_{t-1})I(v'D_T^{-1}z_{t-1}\leq m) \middle| \mathcal{F}_{t-1} \right]\\
&= \sum_{t=1}^T \int_{r=0}^{v'D_T^{-1}z_{t-1}I(v'D_T^{-1}z_{t-1}\leq m)} \left[\int_{s=r}^{v'D_T^{-1}z_{t-1}I(v'D_T^{-1}z_{t-1}\leq m)}ds \right] f_{u_{t\tau},t-1}(r)dr\\
&= \sum_{t=1}^T \int_{s=0}^{v'D_T^{-1}z_{t-1}I(v'D_T^{-1}z_{t-1}\leq m)}
 \int_{r=0}^{s}f_{u_{t\tau},t-1}(r)dr ds\\
&= \sum_{t=1}^T \int_{0}^{v'D_T^{-1}z_{t-1}I(v'D_T^{-1}z_{t-1}\leq m)}
f_{u_{t\tau},t-1}(0) s ds+o_p(1)\\
&= \frac{1}{2} \sum_{t=1}^T f_{u_{t\tau},t-1}(0) (v'D_T^{-1}z_{t-1})^2 I(v'D_T^{-1}z_{t-1}\leq m)+o_p(1).
\end{align*}
Define $\zeta_{t\tau}:=f_{u_{t\tau},t-1}(0)-f_{u_{\tau}}(0)$, then we can write $\overline{W}_{Tm}(v)= A + B + o_p(1)$,
where
\begin{align*}
A &:= \frac{f_{u_{\tau}}(0)}{2}  \sum_{t=1}^T  (v'D_T^{-1}z_{t-1})^2 I(v'D_T^{-1}z_{t-1}\leq m), \quad
B := \frac{1}{2} \sum_{t=1}^T \zeta_{t\tau} (v'D_T^{-1}z_{t-1})^2 I(v'D_T^{-1}z_{t-1}\leq m).
\end{align*}
It follows from \eqref{cgce_1} that
\[
A \Rightarrow \frac{f_{u_{\tau}}(0)}{2} v'\left[ \int \overline{J}_c(s)\overline{J}_c(s)' I(0<v'\overline{J}_c(s)\leq m)ds \right]v:=\eta_m,
\]
where $\overline{J}_c(r):=(1,J_c(r))'$. For $B$, observe that $B$ is bounded by
\begin{equation} \label{B_bound}
\left|\frac{1}{2T} \sum_{t=1}^T \zeta_{t\tau} \left(v_1 + v_2 T^{-1/2}x_{t-1}\right)^2 \right|.
\end{equation}
It follows from Lemma \ref{lemma_joint} that \eqref{B_bound} is $O_p(T^{-1/2})$, and $\overline{W}_{Tm}(v) \Rightarrow \eta_m$ follows.

We proceed to show $W_{Tm}(v)\Rightarrow \eta_m$. Because $\max_{1\leq t\leq T}(v'D_T^{-1}z_{t-1})I(0\leq v'D_T^{-1}z_{t-1}\leq m) =o_p(1)$, we have $\sum_{t=1}^T E[\xi_{tm}(v)^2|\mathcal{F}_{t-1}] \leq \max_{1 \leq t \leq T} \{(v'D_T^{-1}z_{t-1})I(0\leq v'D_T^{-1}z_{t-1}\leq m)\} \sum_{t=1}^T \overline{\xi}_{tm}(v) \rightarrow_p 0$. Therefore, $W_{Tm}(v) - \overline{W}_{Tm}(v) = \sum_{t=1}^T [\xi_{tm}(v)-\overline{\xi}_{tm}(v)] = \sum_{t=1}^T \{ \xi_{tm}(v)-E[\xi_{tm}(v)|\mathcal{F}_{t-1}]\} \rightarrow_p 0$, and $W_{Tm}(v)\Rightarrow \eta_m$ follows. Define
\begin{equation} \label{eta_defn}
\eta = \frac{f_{u_{\tau}}(0)}{2} v'\left[\int \overline{J}_c(s) \overline{J}_c(s)' \right]vI(v' \overline{J}_c(s) >0)ds,
\end{equation}
then we have $\eta_m \Rightarrow \eta$ as $m\rightarrow \infty$. For any $\epsilon>0$, we have $\lim_{m\rightarrow\infty}\limsup_{n\rightarrow\infty}\Pr[|W_T(v)-W_{Tm}(v)|\geq\epsilon]=0$ because $\Pr[|W_T(v)-W_{Tm}(v)|\geq 0]  \leq \Pr[\sum_{t=1}^T I(v'D_T^{-1}z_{t-1}>m) \geq 0]  \leq \Pr[\max_{1\leq t \leq T}\{v'D_T^{-1}z_{t-1}\}>m] \to 0$ as $m \to \infty$ from $\lim_{m\rightarrow\infty}\Pr[\sup_{0\leq r\leq 1}v'\overline{J}_c(r)>m] =0$. Consequently, we have $W_T(v) \Rightarrow \eta$. Similarly, we can show that the third term on the right hand side of \eqref{ZT} converges to $(f_{u_{\tau}}(0)/2) v'\int \overline{J}_c \overline{J}_c' vI(v' \overline{J}_c <0)$ in distribution. Therefore, 
\begin{equation*}
Z_T(v) \Rightarrow
- v'
\int \overline{J}_c d B_{\psi}
 + \frac{f_{u_{\tau}}(0)}{2} v'
\int \overline{J}_c \overline{J}_c'
v := Z(v).
\end{equation*}
Observe that $Z_T(v)$ is minimized at $\widehat{v}=D_T(\widehat{\gamma}(\tau)-\gamma(\tau))$ and $Z(v)$ is minimized at 
\begin{equation*}
\frac{1}{f_{u_{\tau}}(0)}
\left[
\int \overline{J}_c \overline{J}_c'
\right]
^{-1}
\int \overline{J}_c d B_{\psi}.
\end{equation*}
Hence, from Lemma A of \citet{Knight1989}, we have 
\begin{equation*}
D_T(\widehat{\gamma}(\tau)-\gamma(\tau)) \Rightarrow 
\frac{1}{f_{u_\tau}(0)}
\left[
\int \overline{J}_c \overline{J}_c'
\right]
^{-1}
\int \overline{J}_c d B_{\psi},
\end{equation*}
giving the stated result.
\qedsymbol

\subsection{Proof of Proposition \ref{prop_asy_std_t}}

We suppress $(\tau)$ from $\widehat \Delta_{fz}(\tau)$ and $\widehat \Omega (\tau)$. The stated result holds if $T \text{se}(\widehat{\gamma}_1) \Rightarrow (\omega_{\psi}/f_{u_\tau}(0))  [\int (J_c^\mu)^2]^{-1/2}$, which holds if
\begin{equation} \label{hc_cov_limit}
(a)\ T D_T^{-1}\widehat \Delta_{fz}D_T^{-1}  \Rightarrow f_{u_\tau}(0) \int \overline J_c \overline J_c',  \quad
(b)\ T D_T^{-1}\widehat \Omega D_T^{-1} \Rightarrow \int \overline{J}_c \bar{J}_c' \omega_{\psi}^2 .
\end{equation}
We first show (a) of (\ref{hc_cov_limit}). For $\theta \in \mathbb{R}^2$, define $\Delta_{fz}(\theta )$ similarly to $\widehat \Delta_{fz}$ but using $u_{\theta t}:=y_t - \theta' z_{t-1}$ in place of $\widehat{u}_{t\tau}$. Define $\Delta_{fz}$ similarly to $\widehat \Delta_{fz}$ but using $u_{t\tau}$ in place of $\widehat{u}_{t\tau}$. Define $\Theta_T:=\{\theta \in \mathbb{R}^2: \|D_T (\theta- \gamma(\tau)) \| \leq \log T\}$. In view of the convergence rate of $\widehat \gamma(\tau)$, (a) of (\ref{hc_cov_limit}) holds if we show
\begin{align} 
& \sup_{\theta \in \Theta_T}\left\|T  D_T^{-1} \Delta_{fz}(\theta)D_T^{-1} -T  D_T^{-1} \Delta_{fz}D_T^{-1} \right\| = o_p(1), \label{Delta_cgce1} \\
& T D_T^{-1} \Delta_{fz}D_T^{-1} \Rightarrow f_{u_\tau}(0) \int \overline J_c \overline J_c' . \label{Delta_cgce2}
\end{align}
We show (\ref{Delta_cgce1}). For brevity, we focus on the $(2,2)$th element of the term inside $\|\cdot\|$ in (\ref{Delta_cgce1}), namely
\begin{equation}  \label{Delta_cgce12}
\sup_{\theta \in \Theta_T} \left| \frac{1}{T^2 h} \sum_{t=1}^T \left[ \phi \left(\frac{u_{\theta t} }{h} \right) - \phi \left(\frac{u_{t\tau}}{h} \right) \right]x_{t-1}^2 \right|= o_p(1) .
\end{equation}
The other elements are analyzed similarly. A routine calculation gives $u_{\theta t} = u_{t\tau} + (\theta - \gamma(\tau))' z_{t-1}$. When $\theta \in \Theta_T$, we have $|(\theta - \gamma(\tau))' z_{t-1}| \leq |(D_T (\theta - \gamma(\tau)))' D_T^{-1} z_{t-1}| \leq T^{-1/2} \log T \|T^{1/2} D_T^{-1} z_{t-1}\|$. Because $T^{-1/2} \log T/h \to 0$ and $\max_{1\leq t \leq T}\|T^{1/2} D_T^{-1} z_{t-1}\|=O_p(1)$, we obtain $\max_{1\leq t \leq T}|(\theta - \gamma(\tau))' z_{t-1}|/h = o_p(1)$. Therefore, (\ref{Delta_cgce12}) holds if we show, for any $\epsilon \in (0,1)$,
\begin{equation}  \label{hat_f_f3}
E\left[\sup_{|x|\leq \epsilon} \frac{1}{h}  \left| \phi \left(\frac{u_{t\tau}}{h} - x \right) - \phi \left(\frac{u_{t\tau}}{h} \right) \right| \middle|\mathcal{F}_{t-1}\right] \leq \mathcal{C} h + \epsilon a_t,
\end{equation}
where $\mathcal{C}$ and $a_t$ do not depend on $\epsilon$ and $Ea_t^2 = O(1)$. Write the left hand side of (\ref{hat_f_f3}) as
\[
\int \sup_{|x|\leq \epsilon}  \frac{1}{h}\left| \phi \left(\frac{u}{h} - x \right) - \phi \left(\frac{u}{h} \right) \right| f_{u_{t \tau},t-1}(u) du.
\]
Split the integral into two, $\{u:|u|\geq h^{1/2}\}$ and $\{u:|u|\leq h^{1/2}\}$. For the first part, because $|u|/h \geq h^{-1/2} \to \infty$, we obtain, when $h$ is sufficiently small,
\begin{align*}
& \int_{|u| \geq h^{1/2}} \sup_{|x|\leq \epsilon}  \frac{1}{h} \left| \phi \left(\frac{u}{h} - x \right) - \phi \left(\frac{u}{h} \right) \right| f_{u_{t \tau},t-1}(u) du \leq \frac{2}{h} \phi \left(\frac{h^{-1/2}}{2} \right) \int f_{u_{t \tau},t-1}(u) du.
\end{align*}
The right hand side is no larger than $\mathcal{C}h$ from Assumption \ref{kernel}(a). For the second part, write the integral as 
\begin{align*}
& \int_{|u| \leq h^{1/2}}\sup_{|x|\leq \epsilon}  \frac{1}{h} \left| \phi \left(\frac{u}{h} - x \right) - \phi \left(\frac{u}{h} \right) \right| f_{u_{t \tau},t-1}(u) du\\
& \leq \epsilon \int_{|u| \leq h^{1/2}}\sup_{|x|\leq \epsilon}  \frac{1}{h}\left| \phi' \left(\frac{u}{h} - \bar x \right)\right| f_{u_{t \tau},t-1}(u) du, \quad \bar x \in [0,x] \\
& \leq \epsilon \left(\sup_{|u| \leq h^{1/2}} f_{u_{t \tau},t-1}(u) \right) \int \sup_{|x|\leq \epsilon}\left| \phi' \left(s - \bar x\right)\right| ds, \quad s = u/h \\
& \leq \epsilon a_t.
\end{align*}
This gives $\epsilon a_t$ in (\ref{hat_f_f3}). Therefore, (\ref{Delta_cgce12}) is proven, and (\ref{Delta_cgce1}) holds.

We proceed to show (\ref{Delta_cgce2}). Define $\phi_{ht\tau}:= (1/h) \phi (u_{t\tau}/h)$. For brevity, we focus on the $(2,2)$th element of (\ref{Delta_cgce2}) and show
\begin{equation}  \label{Delta_cgce3}
\frac{1}{T^2} \sum_{t=1}^T  x_{t-1}^2 \phi_{ht\tau} \Rightarrow f_{u_\tau}(0) \int J_c^2.
\end{equation}
Let $E^*_{t-1}[ \cdot]$ denote $E[ \cdot | \mathcal{F}_{t-1}^*]$. Observe that 
\begin{align} 
& \frac{1}{T^2} \sum_{t=1}^T  x_{t-1}^2   \left( \phi_{ht\tau} - E^*_{t-1}[\phi_{ht\tau}]\right) =o_p(1),  \label{Delta_cgce31} \\
& \frac{1}{T^2} \sum_{t=1}^T x_{t-1}^2  \left( E^*_{t-1}[\phi_{ht\tau}] - f_{u_{t \tau},t-1}^*(0) \right)  =o_p(1). \label{Delta_cgce32}
\end{align}
(\ref{Delta_cgce31}) holds because $v_t := x_{t-1}^2  \left( \phi_{ht\tau} - E^*_{t-1}[\phi_{ht\tau}] \right)$ satisfies $v_t\in \mathcal{F}^*_t$ and $E_{t-1}^*[v_t]=0$, and hence the left hand side of (\ref{Delta_cgce31}) has mean 0 and variance $T^{-4} \sum_{t=1}^T E(v_t^2) = O(1/(Th))$. (\ref{Delta_cgce32}) holds because
\[
E^*_{t-1}[\phi_{ht\tau}] = \int  \frac{1}{h} \phi\left(\frac{u}{h} \right)  f_{u_{t \tau},t-1}^*(u) du =\int \phi\left(x\right)  f_{u_{t \tau},t-1}^*(xh) dx \to f_{u_{t \tau},t-1}^*(0) .
\]
Finally, $T^{-2} \sum_{t=1}^T x_{t-1}^2  f_{u_{t \tau},t-1}^*(0) =T^{-2} \sum_{t=1}^T x_{t-1}^2 \{ f_{u_{t \tau},t-1}^*(0)-E[f_{u_{t \tau},t-1}^*(0)] \} \\ +T^{-2} \sum_{t=1}^T x_{t-1}^2  E[f_{u_{t \tau},t-1}^*(0)]$, and the first term on the right hand side is $O_p(T^{-1/2})$ from Theorem 4.2 of \citet{Hansen:92}, and the second term converges to $f_{u_\tau}(0) \int J_c^2$ from $f_{u_{ \tau}}(0) = E[f_{u_{t \tau},t-1}^*(0)]$. Therefore, we show (\ref{Delta_cgce3}), and (a) of (\ref{hc_cov_limit}) is proven.

We show (b) of (\ref{hc_cov_limit}). Define $\gamma_u(\ell):= E[ \psi_\tau( u_{(t+\ell)\tau}) \psi_\tau( u_{t\tau})  | \mathcal{F}_{t-1+\ell} ]$ and
\[
\Omega^0:=\sum_{\ell =-m}^m k \left(\frac{ \ell}{m} \right) \Gamma^0(\ell), \quad \Gamma^0(\ell):=  \frac{1}{T} \sum_{t=1}^{T}  z_{t-1}z_{t-1}'\gamma_u(\ell) .
\]
(b) of (\ref{hc_cov_limit}) holds if we show 
\begin{equation} \label{omega_approx}
T D_T^{-1}\widehat \Omega D_T^{-1} - T D_T^{-1}\Omega^0 D_T^{-1} = o_p(1),
\end{equation}
because, in view of $ \omega_{\psi}^2 = \sum_{\ell =-\infty}^\infty \gamma_u (\ell)$,
\[
T D_T^{-1}\Omega^0 D_T^{-1} = T D_T^{-1}\sum_{t=1}^{T} z_{t-1}z_{t-1}'D_T^{-1} \sum_{\ell =-m}^m k \left(\frac{ \ell}{m} \right) \gamma_u(\ell) \Rightarrow \int \overline{J}_c \bar{J}_c' \omega_{\psi}^2 .
\]
For $\theta \in \mathbb{R}^2$, define $u_{\theta t}$ and $\Theta_T$ as in the proof of (a) of (\ref{hc_cov_limit}). Define $\widehat \Gamma(\ell,\theta )$ similarly to $\widehat \Gamma(\ell)$ but using $u_{\theta t}$ in place of $\widehat{u}_{t\tau}$. In view of Assumption \ref{kernel}(b)--(d) and the convergence rate of $\widehat{\gamma}(\tau)-\gamma(\tau)$ shown in Proposition \ref{prop_gamma}, (\ref{omega_approx}) holds if we show
\begin{equation} \label{gamma_approx}
E\sup_{\theta \in \Theta_T}\left\|T  D_T^{-1} \widehat \Gamma(\ell,\theta )D_T^{-1} - T D_T^{-1}\Gamma^0(\ell ) D_T^{-1} \right\| = O(T^{-1}\ell^{1/2} + T^{-1/2}\log T).
\end{equation}

We proceed to show (\ref{gamma_approx}). Without loss of generality, assume $\ell \geq 0$. The case with $\ell<0$ can be analyzed similarly. For brevity, we focus on the $(2,2)$th element of the term inside $\|\cdot\|$ in (\ref{gamma_approx}), namely, 
\begin{equation} \label{gamma22_approx}
E\sup_{\theta \in \Theta_T}T^{-1}\left|  \widehat \Gamma(\ell,\theta )_{22} - \Gamma^0(\ell )_{22}\right| = O(T^{-1}\ell^{1/2} + T^{-1/2} \log T).
\end{equation}
where $\widehat \Gamma(\ell,\theta )_{22}:= T^{-1} \sum_{t=1}^{T-\ell } x_{t-1+\ell } \psi_\tau( u_{\theta, t+\ell})  x_{t-1} \psi_\tau(u_{\theta t})$ and $\Gamma^0(\ell)_{22} := 
T^{-1} \sum_{t=1}^{T} x_{t-1}^2 \gamma_u(\ell)$. The other elements are analyzed similarly.
Define
\begin{align*}
\ddot \Gamma(\ell,\theta )_{22} &:= 
T^{-1} \sum_{t=1}^{T-\ell } x_{t-1}^2 \psi_\tau( u_{\theta,t+\ell}) \psi_\tau( u_{\theta t}) , \\
\widetilde \Gamma(\ell )_{22} &:= 
T^{-1} \sum_{t=1}^{T-\ell } x_{t-1}^2  \psi_\tau( u_{(t+\ell)\tau}) \psi_\tau( u_{t\tau}) .
\end{align*}
From $\sup_u |\Psi_\tau(u)| \leq 1$, we have $| \widehat \Gamma(\ell,\theta )_{22} - \ddot \Gamma(\ell,\theta )_{22} | \leq T^{-1} \sum_{t=1}^{T-\ell }|x_{t-1+\ell}-x_{t-1}| |x_{t-1}| $. Because $E(x_{t-1+\ell}-x_{t-1})^2 \leq \mathcal{C} \ell$ and $E (T^{-1/2} x_t)^2 < \infty$ from Assumption \ref{mixing}, we obtain
\begin{equation} \label{gamma22_1}
E \sup_{\theta \in \Theta_T} \left| \widehat \Gamma(\ell ,\theta)_{22} - \ddot \Gamma(\ell ,\theta)_{22} \right| \leq \mathcal{C} \ell^{1/2}.
\end{equation}
Because $\sup_u |\Psi_\tau(u)| \leq 1$, we have
\[
|\ddot \Gamma(\ell,\theta )_{22} -\widetilde \Gamma(\ell )_{22}| \leq T^{-1} \sum_{t=1}^{T-\ell }x_{t-1}^2 \left[ |\psi_\tau(u_{\theta, t+\ell})-\psi_\tau(u_{(t+\ell)\tau})|  + |\psi_\tau( u_{\theta t}) - \psi_\tau(u_{t\tau})| \right]
\]
In view of $u_{\theta t} = u_{t\tau} + (\theta - \gamma(\tau))' z_{t-1}$, we have $| \psi_\tau( u_{\theta t}) - \psi_\tau(u_{t\tau}) | \leq 1\{ |u_{t\tau}|< |(\theta - \gamma(\tau))' z_{t-1}|\}$. It follows that 
\begin{align*}
E[ x_{t-1}^2 |\psi_\tau( u_{\theta t}) - \psi_\tau(u_{t\tau})|] & = E\left[ x_{t-1}^2 E\left[|\psi_\tau( u_{\theta t}) - \psi_\tau(u_{t\tau})|\middle| \mathcal{F}_{t-1} \right] \right] \\
& \leq E\left[ x_{t-1}^2 \Pr\left( |u_{t\tau}|< |(\theta - \gamma(\tau))' z_{t-1}| \middle| \mathcal{F}_{t-1} \right) \right] \\
& \leq \mathcal{C}E\left[ x_{t-1}^2|(\theta - \gamma(\tau))' z_{t-1}| \right],
\end{align*}
where the last inequality follows from Assumption \ref{error-dist1}. When $\theta \in \Theta_T$, we have $|(\theta - \gamma(\tau))' z_{t-1}| \leq \log T \|D_T^{-1} z_{t-1}\|$, and hence the right hand side is bounded by $\mathcal{C} T^{1/2}\log T$. The term involving $|\psi_\tau(u_{\theta, t+\ell})-\psi_\tau(u_{(t+\ell)\tau})|$ is bounded similarly, and we obtain
\begin{equation} \label{gamma22_2}
E \sup_{\theta \in \Theta_T} \left| \ddot \Gamma(\ell,\theta )_{22} -\widetilde \Gamma(\ell )_{22} \right| \leq \mathcal{C} T^{1/2} \log T.
\end{equation}
Therefore, (\ref{gamma22_approx}) follows from (\ref{gamma22_1}) and (\ref{gamma22_2}). Hence, (b) of (\ref{hc_cov_limit}) is proven. \qedsymbol

\subsection{Proof of Proposition \ref{prop_asy_FM} and Corollary \ref{corollary_FM}}

The proof is similar to the proof of Theorems 2 and 3 of \citet{xiao:09}. For part (a), it follows from Proposition \ref{prop_gamma} and \eqref{asy-breakdown} that
\[
T \widehat{\gamma}_1(\tau) \Rightarrow \frac{\omega^{-1}\delta_{\tau} \sqrt{\tau(1-\tau)}}{
f_{u_\tau}(0)}
\frac{\int J_c^\mu dB_{v}}{\int (J_c^\mu)^2}
+\frac{\omega_{\psi.v}}{f_{u_\tau}(0)}
\frac{\int  J_c^\mu dB_{\psi . v}}{\int (J_c^\mu)^2}.
\]
Therefore, part (a) follows from the definition of $\widehat{\gamma}_1(\tau,c)^{+}$ given in (\ref{eq:gamma:+}), consistency of $(\widehat{\omega}, \widehat{\delta}_{\tau}, \widehat{f_{u_\tau}(0)}, \widehat{\lambda}_{vv})$, (\ref{cgce_1}), and $T^{-1} \sum_{t=1}^T x_{t-1}^{\mu} (x_t -\phi x_{t-1}) \Rightarrow \int J_c^\mu dB_{v} + \lambda_{vv}$ \citep[][Theorem 4.4]{Hansen:92}, and the continuous mapping theorem. Part (b) follows immediately from part (a). Corollary \ref{corollary_FM} follows immediately from Proposition \ref{prop_asy_FM}. \qedsymbol

\subsection{Proof of Proposition \ref{prop_switch1}}

We focus on the proof of part (a). Parts (b) and (c) are proven similarly. When $x_t$ follows \eqref{x_stationary}, the stated result follows from Lemma \ref{lemma:stationary}. Henceforth, we assume $x_t$ follows \eqref{x}. Observe that
\begin{equation}
\Pr \left( \overline \varphi(\tau, \alpha_1, \alpha_2, \overline c_L)=1 \right) 
\leq \Pr \left( \overline\varphi(\tau, \alpha_1, \alpha_2, \overline c_L)=1 \cap c \in [\underline c, \overline c] \right) + \Pr \left(  c \notin [\underline c, \overline c] \right). \label{pbound_1}
\end{equation}
The second probability on the right hand side of (\ref{pbound_1}) is no larger than $ \alpha_1$. For the first probability on the right hand side of (\ref{pbound_1}), consider the case $c \geq \overline c_L$ first. Because $c \leq \overline c$, we have $\overline c_L \leq \overline c$. Therefore, $\overline\varphi(\tau, \alpha_1, \alpha_2, \overline c_L) \leq \overline\varphi_{FM}(\tau,\alpha_1,\alpha_2)$ holds from (\ref{eq:CI:gamma3}), and the first probability on the right hand side of (\ref{pbound_1}) is bounded by $\Pr \left( \overline\varphi_{FM}(\tau,\alpha_1,\alpha_2)=1 \cap c \in [\underline c, \overline c] \right)$. Observe that, when $c \in [\underline c, \overline c]$, we have $\min_{c^*\in \mbox{CI}_c(\alpha_1)} t_{\gamma_1}(\tau,c^*)^+ \leq t_{\gamma_1}(\tau,c)^+$. Therefore, $\Pr \left( \overline\varphi_{FM}(\tau,\alpha_1,\alpha_2)=1 \cap c \in [\underline c, \overline c] \right) \leq \\  \Pr\left( t_{\gamma_1}(\tau,c)^+ \geq z_{\alpha_2/2} \right) \to \alpha_2/2$, and the stated result follows.

Next, consider the case $c < \overline c_L$. Because $\underline c \leq c$, we have $\underline c < \overline c_L$. Therefore, $\overline\varphi(\tau, \alpha_1, \alpha_2, \overline c_L) \leq \overline\varphi_{t}(\tau,\alpha_2,\overline c_L)$ holds from (\ref{eq:CI:gamma3}). Consequently, the first probability on the right hand side of (\ref{pbound_1}) is bounded by $\Pr (t_{\gamma_1}(\tau) \geq z_{1-\alpha_2/2}(\overline c_L))$. Because $T\text{se}(\widehat \gamma_1) \Rightarrow (1/f_{u_\tau}(0)) [\tau(1-\tau)/\int (J_c ^\mu )^2]^{1/2}$ from a standard argument, this probability converges to $\Pr ( Z(c,\delta_{\tau}) \geq z_{1-\alpha_2/2}(\overline c_L) )$. Because $c < \overline c_L$, this is no larger than $\alpha_2/2$, and the stated result follows.  \qedsymbol

\subsection{Proof of Proposition \ref{switch_size}}

We focus on the right-tailed test of $H_0:\gamma_1(\tau)=0$ against $H_A:\gamma_1(\tau)>0$. The left-tailed test is proven similarly. When $x_t$ follows \eqref{x_stationary}, the stated result follows from Lemma \ref{lemma:stationary}. Henceforth, assume $x_t$ follows \eqref{x}. For $c \in Q_c$, the stated result holds from the definition of $\overline\alpha_1$. We show that the stated result holds when $c$ is fixed and sufficiently large negative. Observe that
\begin{align}
 \Pr \left( \overline\varphi(\tau, \overline\alpha_1, \tilde \alpha_2, \overline c_L)=1 \right) 
& \leq \Pr \left( \overline\varphi(\tau, \overline\alpha_1, \tilde \alpha_2, \overline c_L)=1 \cap \underline c \leq \overline c_L \right) + \Pr \left(  \underline c > \overline c_L\right) \nonumber \\
& \leq \Pr \left( \overline\varphi_{t}(\tau,\tilde\alpha_2,\overline c_L)=1 \right) + \Pr \left(  \underline c > \overline c_L\right), \label{p_a1}
\end{align}
where the second inequality follows from (\ref{eq:CI:gamma3}).

For the first probability in (\ref{p_a1}), recall that $t_{\gamma_1}(\tau) \to_d Z(c,\delta_{\tau})$ as $T \to \infty$ from Proposition \ref{prop_asy_std_t} and $Z(c,\delta_{\tau}) \to_d N(0,1)$ as $c \to -\infty$. Therefore, as $c \to -\infty$,
\[
\limsup_{T \to \infty}\Pr \left( \overline\varphi_{t}(\tau,\tilde\alpha_2,\overline c_L)=1 \right) \to \Pr \left( N(0,1) \geq z_{1-\tilde\alpha_2/2}(c_L) \right) \leq \tilde\alpha_2/2 = \alpha_2/2 - \epsilon/2.
\]

For the second probability in (\ref{p_a1}), it follows from equation (7) and Section 5 of \citet{Phillips14} that $\underline c = -2(\hat \tau^2 - z_{\alpha/2} \hat \tau) +O_p(1)$, where $\hat \tau$ is the $t$-statistic for testing $H_0:\phi=1$. Because $\hat \tau \sim N(-|c/2|^{1/2}, 1/4) +O_p(|c|^{-1/2})$ from Theorem 1 and Section 5 of \citet{Phillips14}, the second probability in (\ref{p_a1}) tends to 0 as $c \to -\infty$, and the stated result follows. \qedsymbol

\subsection{Proof of Proposition \ref{prop_local}}

The proof follows the proof of Proposition \ref{prop_gamma}. Assume $\zeta=0$ without loss of generality. Define $D_T^\kappa := \text{diag}(T^{1/2-\kappa},T^{1-\kappa})$ and $b_{\tau}:=b(Q_{e_t}(\tau))$. Define $\ve_{t\tau} := y_t - \gamma_0(\tau) - \gamma_1 x_{t-1} = e_{t\tau} + T^{\kappa-1} b(e_t)|x_{t-1}|$ and define $\theta(\tau) := (\gamma_0(\tau),\gamma_1)'$, so that $\ve_{t\tau} = y_t - \theta(\tau)'z_{t-1}$. Note that minimization \eqref{qr_estimator} is equivalent to $\min_v Z_T^\kappa(v)$, where $Z_T^\kappa(v) := T^{-2\kappa}\sum_{t=1}^T [ \rho_\tau(\ve_{t\tau}- (D_T^\kappa)^{-1} v)'z_{t-1}) - \rho_\tau(\ve_{t\tau})]$. If $\widehat v$ is a minimizer of $Z_T^\kappa(v)$, then we have $\widehat v = D_T^\kappa(\widehat{\theta}(\tau)-\theta(\tau))$. The stated result holds if the finite-dimensional distributions of $Z_T^\kappa (\cdot)$ converge weakly to those of $Z^\kappa (\cdot)$ and $Z^\kappa (\cdot)$ has a unique minimum.

As in the proof of Proposition \ref{prop_gamma}, we have
\begin{equation} \label{GT_alpha}
\begin{aligned}
Z_T^\kappa(v) &= -T^{-2\kappa}\sum_{t=1}^T ((D_T^\kappa)^{-1} v)'  z_{t-1}\psi_\tau(\ve_{t\tau}) \\
& \quad - T^{-2\kappa}\sum_{t=1}^T (\ve_{t\tau}-((D_T^\kappa)^{-1} v)'z_{t-1}) I(0<\ve_{t\tau}<((D_T^\kappa)^{-1} v)'z_{t-1})\\
& \quad +T^{-2\kappa}\sum_{t=1}^T (\ve_{t\tau}-((D_T^\kappa)^{-1} v)'z_{t-1}) I(0>\ve_{t\tau}>((D_T^\kappa)^{-1} v)'z_{t-1}) .
\end{aligned}
\end{equation}
For the first term on the right hand side of \eqref{GT}, it follows from Lemma \ref{lemma_local} that
\[
T^{-2\kappa}(D_T^{\kappa})^{-1}\sum_{t=1}^T z_{t-1}\psi_\tau(\ve_{t\tau})  \Rightarrow 
b_{\tau} f_{e_{\tau}}(0)
\begin{bmatrix}
\int |J_c|\\
\int J_c |J_c|
\end{bmatrix}. 
\]

Similar to the proof of Proposition \ref{prop_gamma}, write the second term on the right hand side of \eqref{GT_alpha} as $W_T^\kappa(v) :=\sum_{t=1}^T \xi_t^\kappa(v)$,  where $\xi_t^\kappa(v)  = T^{-2\kappa}(v'(D_T^\kappa)^{-1}z_{t-1}-\ve_{t\tau})I(0<\ve_{t\tau}<v' (D_T^\kappa)^{-1} z_{t-1})$, and define $W_{Tm}^\kappa(v):=\sum_{t=1}^T \xi_{tm}^\kappa(v)$, where $\xi_{tm}^\kappa(v):=\xi_{t}^\kappa(v)I(v' (D_T^\kappa)^{-1} z_{t-1}\leq m T^{\kappa})$. Denote $\overline{\xi}_{tm}^\kappa(v):=E[\xi_{tm}^\kappa (v)|\mathcal{F}_{t-1}]$ and $\overline{W}_{Tm}^\kappa(v):=\sum_{t=1}^T \overline{\xi}_{tm}^\kappa(v)$. Because $z_{t-1}\in \mathcal{F}_{t-1}$, 
we have, similar to the proof of Proposition \ref{prop_gamma},
\begin{equation*}
\overline{W}_{Tm}^{\kappa}(v) = T^{-2\kappa}\sum_{t=1}^T \int_{s=0}^{v'(D_T^\kappa)^{-1}z_{t-1}I(v'(D_T^\kappa)^{-1}z_{t-1}\leq mT^{\kappa})} \int_{r=0}^{s} f_{\ve_{t\tau,t-1}}(r)dr ds.
\end{equation*}
Because $f_{\ve_{t\tau,t-1}}(\cdot) = f_{e_{\tau}}(\cdot)  + o_p(1)$ around 0 from Lemma \ref{lemma_density}(b) and $\max_{1\leq t\leq T}(v'(D_T^\kappa)^{-1}z_{t-1})I(0\leq v'(D_T^\kappa)^{-1}z_{t-1}\leq m T^{\kappa}) =o_p(1)$, we can write the right hand side as
\begin{align*}
& T^{-2\kappa}\sum_{t=1}^T \int_{0}^{v' (D_T^{\kappa})^{-1}z_{t-1}I(v'(D_T^\kappa)^{-1}z_{t-1}\leq mT^{\kappa})}
f_{e_{\tau}}(0)ds+o_p(1)\\
&= T^{-2\kappa}\frac{1}{2} \sum_{t=1}^T f_{e_{t\tau}}(0) (v'(D_T^\kappa)^{-1}z_{t-1})^2 I(v'(D_T^\kappa)^{-1}z_{t-1}\leq mT^{\kappa})+o_p(1) \\
&= \frac{1}{2} \sum_{t=1}^T f_{e_{\tau}}(0) (v'D_T^{-1}z_{t-1})^2 I(v' D_T^{-1}z_{t-1}\leq m )+o_p(1).
\end{align*}
It follows from \eqref{cgce_1} that the right hand side converges in distribution to 
\[
\frac{f_{e_{\tau}}(0)}{2} v'\left[ \int B_z(s) B_z(s)' v I(0<v'B_z(s)\leq m)\right]ds :=\eta_m,
\]
where $B_z(r):=(1,J_c(r))'$, and $\overline{W}_{Tm}^{\kappa}(v) \Rightarrow \eta_m$ follows. Using a similar argument to the proof of Proposition \ref{prop_gamma} gives $W_{Tm}(u)\Rightarrow \eta_m$. Define
\begin{equation} \label{eta_defn}
\eta = \frac{f_{e_{\tau}}(0)}{2} v'\left[\int B_z(s) B_z(s)' I(v'B_z(s) >0)ds\right]v,
\end{equation}
Using a similar argument to the proof of Proposition \ref{prop_gamma} gives  $W_T(v) \Rightarrow \eta$. Similarly, we can show that the third term on the right hand side of \eqref{GT_alpha} converges to $(f_{e_{\tau}}(0)/2) v'\int B_z B_z' vI(v'B_z<0)$ in distribution. Therefore, 
\begin{equation*}
Z_T^{\kappa}(v) \Rightarrow
- b_{\tau} f_{e_{\tau}}(0) v'
\begin{bmatrix}
\int |J_c|\\
\int J_c |J_c| 
\end{bmatrix}
 + \frac{f_{e_{\tau}}(0)}{2} v'
\begin{bmatrix}
1&\int J_c\\
\int J_c & \int J_c^2\\
\end{bmatrix}
v := Z^{\kappa}(v).
\end{equation*}
Observe that $Z_T^{\kappa}(v)$ is minimized at $\widehat{v}=D_T^{\kappa}(\widehat{\gamma}(\tau)-\gamma(\tau))$ and $Z^{\kappa}(v)$ is minimized at 
\begin{equation*}
b_{\tau}
\begin{bmatrix}
1&\int J_c\\
\int J_c & \int J_c^2\\
\end{bmatrix}
^{-1}
\begin{bmatrix}
\int |J_c|\\
\int J_c |J_c| 
\end{bmatrix}.
\end{equation*}
Therefore, the stated result follows from Lemma A of \citet{Knight1989}.
\qedsymbol

\subsection{Proof of Proposition \ref{switch_power}}

Without loss of generality, let $\zeta=0$. We prove the stated result by showing that, for any $M \in (0,\infty)$,
\begin{equation}\label{divergence}
(a) \quad  \Pr \left( \left| t_{\gamma_1}(\tau) \right|> M \right) \to 1,\quad (b) \quad  \Pr \left(\min_{c^*\in \mbox{CI}_c(\alpha_1)}\left|  t_{\gamma_1}(\tau,c^*)^+ \right| > M \right) \to 1.
\end{equation}

In view of Proposition \ref{prop_local}, part (a) of \eqref{divergence} holds if $T D_T^{-1}\widehat \Delta_{fz}D_T^{-1}  \Rightarrow f_{e_\tau}(0) \int \overline J_c \overline J_c'$ and $T D_T^{-1}\widehat \Sigma(\tau) D_T^{-1} \Rightarrow \int \overline J_c \overline J_c' \omega_{\psi}^2$. Because $e_t - Q_{e_t}(\tau)$ dominates $u_{t\tau}$, these results follow from the proof of Proposition \ref{prop_asy_std_t} and Lemma \ref{lemma_density}(a).

We proceed to show part (b) of \eqref{divergence}. Note that the consistency of $\widehat{f_{u_\tau}(0)}$, $\widehat{\omega}_\psi$, $\widehat{\omega}_v$, $\widehat{\delta}_{\tau}$, and $\widehat{\lambda}_{vv}$ follows from a standard argument and the $(1,1)$th element of (b) of (\ref{hc_cov_limit}). Therefore, it follows from (\ref{eq:gamma:+}) that
\begin{align*}
T^{1-\kappa}\widehat{\gamma}_1(\tau,c^*)^{+} & = T^{1-\kappa}\widehat{\gamma}_1(\tau,c)^{+} + 
T^{1-\kappa}\frac{\widehat{\omega}_\psi\widehat{\omega}_v^{-1} \widehat{\delta}_{\tau} }{\widehat{f_{u_\tau}(0)}  }
  \left( \frac{c^* - c}{T}\right) = T^{1-\kappa}\widehat{\gamma}_1(\tau) + o_p(1).
\end{align*}
Because $T \mbox{se}(\widehat{\gamma}_1(\tau,c^*)^{+}) \to_d ({\omega}_{\psi.v} / {f_{u_\tau}(0)}) (\int (J_c^{\mu})^2 )^{-1/2}$, part (b) of \eqref{divergence} follows. \qedsymbol

\section{Auxiliary Results}

The following lemma establishes weak convergence of partial sums of the variables that appear in the paper. Define $\zeta_{t\tau}:=f_{u_{t\tau},t-1}(0)-f_{u_{\tau}}(0)$.
\begin{lemma} \label{lemma_joint}
Under the assumptions of Proposition \ref{prop_gamma}, we have
\begin{equation*}
(a)\  \sum_{t=1}^{T}\zeta_{t\tau} =O_p \left(T^{1/2} \right), \quad (b)\  \sum_{t=1}^T \zeta_{t\tau}x_{t-1} = O_p \left( T \right),\quad
(c)\ \sum_{t=1}^T \zeta_{t\tau}x_{t-1}^2 = O_p \left( T^{3/2} \right).
\end{equation*}
\end{lemma}
\begin{proof}[Proof of Lemma \ref{lemma_joint}]
Parts (a) and (b) follow from Assumption \ref{mixing}, Theorem 4.4 of \citet{Hansen:92}, and continuous mapping theorem. For part (c), Theorem 4.2 of \citet{Hansen:92} derives the limit of $T^{-3/2}\sum_{t=1}^T \zeta_{t\tau}V_{t-1}^2$, where $V_t = \sum_{j=1}^t v_j$ is an $I(1)$ process. Part (c) then follows from adjusting the proof of Theorem 4.2 of \citet{Hansen:92} by applying the proof of Theorem 4.4 of \citet{Hansen:92}.
\end{proof}
The following lemma shows that, when $x_t$ is stationary, the asymptotic size of the one-tailed switching-FM test is no larger than $\alpha_2/2$ for any first-stage confidence level $\alpha_1$.

\begin{lemma}\label{lemma:stationary}
Suppose $x_t$ follows \eqref{x_stationary}. Then, for any $\alpha_1 \in (0,1)$, we have
\begin{align*}
\limsup_{T \to \infty} \Pr \left( \overline \varphi(\tau, \alpha_1, \alpha_2, \overline c_L)=1 \right) \leq \alpha_2/2, \\
\limsup_{T \to \infty} \Pr \left( \underline\varphi(\tau, \alpha_1, \alpha_2, \underline c_L)=1 \right) \leq \alpha_2/2 . 
\end{align*}
\end{lemma}
\begin{proof}[Proof of Lemma \ref{lemma:stationary}]
We prove the first result only. The second result is proven similarly. Observe that 
\begin{align}
 \Pr \left( \overline\varphi(\tau, \alpha_1, \alpha_2, \overline c_L)=1  \right) 
& \leq \Pr \left( \overline\varphi(\tau, \alpha_1, \alpha_2, \overline c_L)=1  \cap \underline c \leq \overline c_L \right) + \Pr \left(  \underline c > \overline c_L\right) \nonumber \\
& \leq \Pr \left( \overline\varphi_{t}(\tau,\alpha_2, \overline c_L)=1\right) + \Pr \left(  \underline c > c_L\right) , \label{p_s1}
\end{align}
where the second line follows from (\ref{eq:CI:gamma3}). As $T \to \infty$, the first probability in (\ref{p_s1}) converges to $p \leq \alpha_2/2$ because $t_{\gamma_1}^{HC}(\tau) \to_d N(0,1)$ from a standard argument and $z_{1-\alpha_2/2}(\overline c_L) > z_{1-\alpha_2/2}$.

For the second probability in (\ref{p_s1}), from Theorem 2 and Section 5 of \citet{Phillips14}, the lower bound of the confidence interval of $\rho$ implied by the $1-\underline\alpha_1$ confidence interval of $c$ is given by $\underline{\rho} = 1-2A_\rho +O_p(T^{-1/2})$, where $A_\rho=(1-\rho)/(1+\rho)>0$. Consequently, the lower bound of the confidence interval of $c$ is $\underline{c}= T(\underline{\rho}-1) = -T2A_p + O_p(T^{1/2})$, and the second probability in (\ref{p_s1}) converges to 0 as $T \to \infty$. Therefore, the stated result holds.
\end{proof}

\begin{lemma} \label{lemma_local}
Suppose the assumptions of Proposition \ref{prop_local} hold. Define $D_T^\kappa := \text{diag}(T^{1/2-\kappa},T^{1-\kappa})$ and $\ve_{t\tau} := e_{t\tau} + T^{\kappa-1} b(e_t)| x_{t-1} +\zeta|$. Then,
\[
T^{-2\kappa}(D_T^{\kappa})^{-1}\sum_{t=1}^T z_{t-1}\psi_\tau(\ve_{t\tau}) \Rightarrow 
          b(Q_{e_t}(\tau))f_{e_{\tau}}(0) \begin{bmatrix}
            \int |J_c|\\
            \int J_c|J_c|
          \end{bmatrix}.
\]
\end{lemma}

\begin{proof}[Proof of Lemma \ref{lemma_local}]
Assume $\zeta=0$ without loss of generality. Let $E_{t-1}[ \cdot]$ denote $E[ \cdot | \mathcal{F}_{t-1} ]$, and let $b_{\tau}$ denote $b(Q_{e_t}(\tau))$. We first show
\begin{equation} \label{z_psi_sum}
\begin{aligned}
& T^{-2\kappa} (D_T^{\kappa})^{-1}\sum_{t=1}^T z_{t-1}\psi_\tau(\ve_{t\tau}) \\
& = T^{-2\kappa} (D_T^{\kappa})^{-1}
\sum_{t=1}^T z_{t-1}E_{t-1}[\psi_\tau(\ve_{t\tau})]
+ T^{-2\kappa} (D_T^{\kappa})^{-1}
\sum_{t=1}^T z_{t-1}\psi_\tau(u_{t\tau}) + o_p(1),
\end{aligned}
\end{equation}
with $u_{t\tau}$ defined in (\ref{uttau_new}) and satisfying $\ve_{t\tau} = u_{t\tau} + T^{\kappa-1} b_{\tau} | x_{t-1} |$. Define
\[
G_T^{\kappa}(\epsilon) := 
T^{-2\kappa} (D_T^{\kappa})^{-1} \sum_{t=1}^T z_{t-1} \left\{\psi_\tau(u_{t\tau} + \epsilon| x_{t-1}|) - E_{t-1}\left[ \psi_\tau(u_{t\tau} + \epsilon| x_{t-1}|) \right] \right\}.
\]
Because $E_{t-1}\left[\psi_\tau(u_{t\tau})\right]=0$, (\ref{z_psi_sum}) holds if, for a generic constant $C>0$, 
\begin{equation} \label{GT}
\sup\left\{ \|G_T^{\kappa}(\epsilon) - G_T^{\kappa}(0)\|: |\epsilon| \leq CT^{\kappa-1} \right\} = o_p(1).
\end{equation}
The proof of (\ref{GT}) follows the proof of Lemma A.3 of \citet{Lee2016}; using the conditional expectation rather than unconditional expectation avoids the nonstationarity problem, and the stochastic equicontinuity proof of \citet{bickel75jasa} with i.i.d.\ regressors can be modified accordingly. In fact, $T^{-2\kappa} (D_T^{\kappa})^{-1} z_{t-1}$ and $\psi_\tau$ satisfy conditions (4.1) and C1 of \citet{bickel75jasa}, respectively. Hence, the analogy of Lemma 4.1 of \citet{bickel75jasa} holds, and (\ref{GT}) follows.

We proceed to evaluate the right hand side of (\ref{z_psi_sum}). Observe that 
\[
u_{t\tau}  = (e_t - Q_{e_t}(\tau)) \left(1 + \frac{b(e_t) - b (Q_{e_t}(\tau))}{ e_t - Q_{e_t}(\tau)} \frac{|x_{t-1}|}{T^{1-\kappa}}\right).
\]
Therefore, $u_{t\tau}>0$ if and only if $e_t - Q_{e_t}(\tau)>0$, and we have $\psi_\tau(u_{t\tau})  = \psi_\tau(e_{t\tau})$. Consequently, the second term on the right hand side of (\ref{z_psi_sum}) is $o_p(1)$ in view of $T^{-2\kappa} (D_T^{\kappa})^{-1} = \text{diag}(T^{-1/2-\kappa},T^{-1-\kappa})$ and $E_{t-1}[\psi_\tau(e_{t\tau})]=E[\psi_\tau(e_{t\tau})]=0$.

For the first term on the right hand side of (\ref{z_psi_sum}), observe that, in view of $E_{t-1}\left[\psi_\tau(u_{t\tau})\right]=0$,
\begin{align*}
E_{t-1}[\psi_\tau(\ve_{t\tau})] & = E_{t-1}[\psi_\tau(u_{t\tau}+ T^{\kappa-1} b_{\tau} |x_{t-1}|)] \\
& = \left. \frac{\partial E_{t-1}[\psi_\tau(u_{t\tau} + \epsilon |x_{t-1}|)]}{\partial \epsilon} \right|_{\epsilon = 0} T^{\kappa-1} b_{\tau} + o_p(T^{\kappa-1}|x_{t-1}|).
\end{align*}
Because $E_{t-1}[\psi_\tau(u_{t\tau} + \epsilon |x_{t-1}|)] = \tau - E_{t-1}[I\{u_{t\tau} < - \epsilon |x_{t-1}|\}] = \tau - \int_{-\infty}^{-\epsilon |x_{t-1}|} f_{u_{t \tau},t-1}(s)ds$, we have $(\partial/ \partial\epsilon) E_{t-1}[\psi_\tau(u_{t\tau} + \epsilon |x_{t-1}|)] |_{\epsilon = 0} = |x_{t-1}| f_{u_{t \tau},t-1}(0)$. Consequently, in view of Lemma \ref{lemma_density}, we can write the first term on the right hand side of (\ref{z_psi_sum}) as
\[
f_{e_{\tau}}(0) b_{\tau}\ \text{diag} (T^{-3/2},T^{-2}) \sum_{t=1}^T z_{t-1}|x_{t-1}| + o_p(1),
\]
and the stated result follows.
\end{proof}

\begin{lemma}\label{lemma_density}
Suppose the assumptions of Proposition \ref{prop_local} hold. Let $u_{t \tau}$ defined in (\ref{uttau_new}). Let $\ve_{t\tau} := e_{t\tau} + T^{\kappa-1} b(e_t)| x_{t-1} +\zeta|$, and let $f_{\ve_{t \tau},t-1}(\cdot)$ denote the density of $\ve_{t\tau}$ conditonal on $\mathcal{F}_{t-1}$. Then, there exists $\eta>0$ such that
\begin{align*}
(a)\ E \sup_{\{s:|s| \leq \eta\}} |f_{u_{t \tau},t-1}(s) - f_{e_\tau}(s) |^2 = o(1),  \quad (b)\ E \sup_{\{s:|s| \leq \eta\}} |f_{\ve_{t \tau},t-1}(s) - f_{e_\tau}(s)|^2 = o(1), 
\end{align*}
\end{lemma}

\begin{proof}[Proof of Lemma \ref{lemma_density}]
Without loss of generality, assume $\zeta=0$ and $s \geq 0$. We prove part (a) first. Observe that, because $u_{t\tau}  = e_t - Q_{e_t}(\tau)  + [b(e_t) - b (Q_{e_t}(\tau))] T^{\kappa-1} |x_{t-1}|$,
\begin{align*}
u_{t\tau} \leq s & \Longleftrightarrow
e_t - Q_{e_t}(\tau)  +  [b(e_t) - b (Q_{e_t}(\tau))]T^{\kappa-1}  |x_{t-1}| \leq s  \\
& \Longleftrightarrow  e_t - Q_{e_t}(\tau)  \leq \left( 1 + \frac{b(e_t) - b (Q_{e_t}(\tau))}{e_t - Q_{e_t}(\tau)} T^{\kappa-1}|x_{t-1}| \right)^{-1}  s .
\end{align*}
Let $P_{t-1}(\cdot)$ denote $\Pr(\cdot|\mathcal{F}_{t-1})$. Because $0 \leq \frac{b(e_t) - b (Q_{e_t}(\tau))}{e_t - Q_{e_t}(\tau)} \leq M_b$ and $(1+a)^{-1} \geq 1-a$ for $a>0$, we have
\begin{equation}\label{p_e_u1}
P_{t-1}\left( e_{t \tau} \leq \left( 1 - M_b T^{\kappa-1}|x_{t-1}| \right)  s \right) \leq  P_{t-1}\left( u_{t\tau} \leq s \right)  \leq  P_{t-1}\left( e_{t \tau} \leq s \right).
\end{equation}
It follows that, in view of the independence between $e_t$ and $\mathcal{F}_{t-1}$,
\begin{equation}\label{p_e_u2}
\sup_{s} \left|P_{t-1}\left( u_{t\tau} \leq s \right)  -  P_{t-1}\left( e_{t \tau} \leq s \right)\right| \leq \left\{\sup_s |s| f_{e_\tau}(s) \right\} M_b T^{\kappa-1}|x_{t-1}|.
\end{equation}
The term in the braces is bounded by $C<\infty$ because $E|e_{t\tau}|< \infty$. Let $\Delta = 1/2-\kappa>0$, so that $T^{\kappa-1} = T^{1/2-\Delta}$. Therefore,
\begin{align*}
&\sup_{s}T^{\Delta/2} \left| \left[ P_{t-1}\left( u_{t\tau} \leq s + T^{-\Delta/2} \right)  - P_{t-1}\left( u_{t\tau} \leq s\right) \right] \right. \\
&- \left. \left[   P_{t-1}\left( e_{t \tau} \leq s + T^{-\Delta/2} \right) - P_{t-1}\left( e_{t \tau} \leq s  \right) \right]\right| \leq 2CM_bT^{1/2-\Delta/2}|x_{t-1}|.
\end{align*}
Letting $T \to \infty$ and noting the independence between $e_t$ and $\mathcal{F}_{t-1}$ give part (a).

For part (b), observe that $\ve_{t\tau} =u_{t\tau} + T^{\kappa-1} b(Q_{e_t}(\tau))|x_{t-1}|$. Therefore, in view of (\ref{p_e_u1}) and (\ref{p_e_u2}), we have
\begin{equation*}
\sup_{s} \left|P_{t-1}\left( u_{t\tau} \leq s \right)  -  P_{t-1}\left( e_{t \tau} \leq s \right)\right| \leq C (M_b + |b(Q_{e_t}(\tau))|)T^{\kappa-1}|x_{t-1}|,
\end{equation*}
and the stated result follows.
\end{proof}

\section{Tables}

\begin{table}[h]
\caption{Percentiles of $Z(c,\delta_\tau)$}\label{z_table}
\centering
\begin{tabular}{r|ccccc|ccccc}
\hline
 &\multicolumn{5}{c|}{$5$ percentile} & \multicolumn{5}{c}{$95$ percentile} \\
 &\multicolumn{5}{c|}{$\delta_\tau$} & \multicolumn{5}{c}{$\delta_\tau$} \\
$c$  & $-1$ & $-0.90$ & $-0.60$ & $-0.30$ & $0.00$& $-1$ & $-0.90$ & $-0.60$ & $-0.30$ & $0.00$  \\
\hline
  $5$  & -0.992 & -1.069 & -1.278 & -1.467 & -1.646 & 2.224 & 2.171 & 1.988 & 1.815 & 1.644 \\
  $0$  &  0.078 & -0.100 & -0.644 & -1.163 & -1.646 & 2.862 & 2.778 & 2.466 & 2.081 & 1.644 \\
$-10$  & -0.927 & -1.004 & -1.230 & -1.445 & -1.646 & 2.240 & 2.188 & 2.019 & 1.836 & 1.644 \\ 
$-20$  & -1.145 & -1.198 & -1.355 & -1.503 & -1.646 & 2.089 & 2.047 & 1.919 & 1.783 & 1.644 \\
$-30$  & -1.242 & -1.283 & -1.411 & -1.531 & -1.646 & 2.014 & 1.979 & 1.871 & 1.759 & 1.644 \\
$-40$  & -1.298 & -1.334 & -1.443 & -1.547 & -1.646 & 1.968 & 1.936 & 1.843 & 1.744 & 1.644 \\
$-50$  & -1.336 & -1.368 & -1.465 & -1.557 & -1.646 & 1.937 & 1.907 & 1.823 & 1.733 & 1.644 \\
$-60$  & -1.363 & -1.393 & -1.481 & -1.564 & -1.646 & 1.910 & 1.884 & 1.808 & 1.726 & 1.644 \\
$-80$  & -1.402 & -1.428 & -1.504 & -1.576 & -1.646 & 1.877 & 1.853 & 1.786 & 1.715 & 1.644 \\
$-100$ & -1.429 & -1.452 & -1.519 & -1.584 & -1.646 & 1.852 & 1.831 & 1.771 & 1.708 & 1.644 \\
$-130$ & -1.456 & -1.476 & -1.535 & -1.593 & -1.646 & 1.826 & 1.809 & 1.756 & 1.700 & 1.644 \\
$-160$ & -1.476 & -1.494 & -1.546 & -1.599 & -1.646 & 1.808 & 1.791 & 1.744 & 1.694 & 1.644 \\
$-190$ & -1.490 & -1.505 & -1.554 & -1.603 & -1.646 & 1.794 & 1.779 & 1.736 & 1.690 & 1.644 \\
\hline
\end{tabular}
\begin{minipage}[l]{ 145  mm}\footnotesize
The percentiles shown above are simulated by approximating $B_v(r)$ by $T^{-1/2}\sum_{t=1}^{[Tr]}v_t$ with $v_t\sim i.i.d.\ N(0,1)$ using $1,000,000$ replications with $T=10,000$.
\end{minipage}
\end{table}

\begin{table}
\caption{First-stage significance levels \label{table:alpha:1:A}}
\centering
\begin{tabular}{|cccc|cccc|}
\hline
   $\delta_{\tau}$ & $\delta $&  $\underline{\alpha}_1$& $\overline{\alpha}_1$
& $ \delta_{\tau}$ &   $\delta $  & $\underline{\alpha}_1$& $\overline{\alpha}_1$ 
\\ \hline
-0.797 &-0.999  &0.14   & 0.43 & -0.399 &-0.500  &0.26   & 0.75 \\
-0.758 &-0.950  &0.15   & 0.50 & -0.359 &-0.450  &0.28   & 0.82 \\
-0.718 &-0.900  &0.17   & 0.51 & -0.319 &-0.400  &0.28   & 0.89 \\
-0.678 &-0.850  &0.18   & 0.56 & -0.279 &-0.350  &0.28   & 0.92 \\
-0.638 &-0.800  &0.19   & 0.58 & -0.239 &-0.300  &0.30   & 0.98 \\
-0.598 &-0.750  &0.20   & 0.62 & -0.199 &-0.250  &0.32   & 0.98 \\
-0.558 &-0.700  &0.21   & 0.65 & -0.159 &-0.200  &0.37   & 0.98 \\
-0.518 &-0.650  &0.22   & 0.68 & -0.119 &-0.150  &0.50   & 0.98 \\
-0.478 &-0.600  &0.23   & 0.70 & -0.080 &-0.100  &0.61   & 0.98 \\
-0.439 &-0.550  &0.24   & 0.73 & -0.040 &-0.050  &0.79   & 0.98 \\
  \hline
\end{tabular}
\\
\begin{minipage}[l]{ 162  mm}\footnotesize
  This table provides the adjusted significance levels, $\underline{\alpha}_1$ and $\overline{\alpha}_1$, respectively, used for the first-stage confidence interval on $c$ when employing the left- and right-sided switching-FM tests at a $5\%$ significance level. They are based on an approximation to the large sample null distribution of the switching-FM test for $\tau=0.5$ using $10,000$ replications with $T=5,000$. We simulate $x_t = (1 - \phi)x_{t-1} + v_t$ with $\phi =  1 + c/T$ and $x_0=0$ and $y_t = e_t$, with $(v_{t},e_{t})$ drawn from an i.i.d.\ bivariate normal distribution with zero means, unit variances, and correlation $\delta$. As detailed in Section  \ref{sec:adjust}, we set $\alpha_2 = 0.1$, $\epsilon = 0.04$,  $\tilde \alpha_2 = \alpha_2 - \epsilon$, and $Q_c=[-120,4]$. For each $\delta$, we select $\overline\alpha_1$ and $\underline\alpha_1$ over the grid $\{0.01, 0.02, \ldots, 0.98\}$ to set the maximum rejection rate over $c\in \{4,\ 2.5,\ 0, -2.5, -5, -10, -15, -25, -35, -40, -50, -60, -80, -100, -120\} \subset Q_c$ as close as possible to  $\tilde{\alpha}/2$ without exceeding it. Columns 1 \& 5 are used to conduct the lookup based on the value of $\delta_{\tau}$.
\end{minipage}
\end{table}

\begin{table} 
\caption{Null rejection rates of the standard quantile, switching-FM, and $t^w$ tests\label{table:size}}
\centering
\newcolumntype{H}{>{\setbox0=\hbox\bgroup}c<{\egroup}@{}}
\begin{tabular}{ccHcHcHcHclcHcHcHcHc}
\hline         
$c\backslash\tau$& 0.100&0.200&0.300&0.400&0.500&0.600&0.700&0.800&0.900
&&0.100&0.200&0.300&0.400&0.500&0.600&0.700&0.800&0.900
\\
\hline
\multicolumn{10}{c}{Panel A: $T=800$, $\delta = -0.95$}&\multicolumn{10}{c}{Panel B: $T=1600$, $\delta = -0.95$}\\

&\multicolumn{9}{c}{Standard quantile $t$-test} 
&&\multicolumn{9}{c}{Standard quantile $t$-test}  \\ 

0&0.215&0.248&0.268&0.273&0.270&0.267&0.260&0.248&0.228&&0.215&0.251&0.262&0.278&0.289&0.278&0.266&0.255&0.218\\
-5&0.133&0.140&0.140&0.141&0.143&0.141&0.138&0.138&0.132&&0.123&0.137&0.142&0.140&0.144&0.145&0.141&0.138&0.127\\
-10&0.110&0.112&0.111&0.111&0.113&0.109&0.110&0.114&0.107&&0.100&0.108&0.108&0.110&0.112&0.111&0.107&0.109&0.098\\
-25&0.086&0.090&0.083&0.083&0.085&0.083&0.082&0.086&0.089&&0.084&0.083&0.084&0.083&0.083&0.082&0.080&0.082&0.077\\
-200&0.070&0.062&0.057&0.053&0.055&0.055&0.059&0.061&0.064&&0.065&0.061&0.060&0.056&0.057&0.057&0.056&0.059&0.064\\

&\multicolumn{9}{c}{Switching-FM test} 
&&\multicolumn{9}{c}{Switching-FM test}  \\ 

 0&0.077&0.063&0.056&0.054&0.052&0.052&0.054&0.064&0.080&&0.066&0.056&0.053&0.051&0.051&0.049&0.055&0.060&0.064\\
-5&0.072&0.059&0.054&0.050&0.047&0.049&0.052&0.062&0.068&&0.062&0.057&0.052&0.049&0.051&0.047&0.048&0.056&0.060\\
-10&0.067&0.058&0.050&0.050&0.050&0.047&0.049&0.059&0.065&&0.056&0.055&0.049&0.045&0.051&0.046&0.046&0.054&0.055\\
-25&0.058&0.050&0.043&0.040&0.041&0.040&0.043&0.051&0.057&&0.049&0.047&0.041&0.039&0.038&0.041&0.040&0.042&0.048\\
-200&0.066&0.060&0.057&0.051&0.056&0.052&0.056&0.060&0.062&&0.066&0.058&0.056&0.056&0.055&0.055&0.054&0.055&0.061\\

&\multicolumn{9}{c}{$t^w$ test} 
&&\multicolumn{9}{c}{$t^w$test} \\

0&0.077&0.073&0.074&0.070&0.070&0.070&0.072&0.077&0.081&&0.073&0.066&0.063&0.064&0.063&0.062&0.064&0.066&0.074\\
-5&0.052&0.044&0.041&0.040&0.042&0.039&0.041&0.050&0.055&&0.045&0.035&0.037&0.037&0.034&0.035&0.037&0.042&0.047\\
-10&0.050&0.044&0.041&0.038&0.036&0.038&0.040&0.042&0.045&&0.039&0.033&0.032&0.032&0.032&0.031&0.034&0.035&0.041\\
-25&0.042&0.041&0.037&0.035&0.030&0.035&0.036&0.041&0.049&&0.038&0.038&0.030&0.030&0.030&0.031&0.030&0.035&0.039\\
-200&0.059&0.050&0.045&0.042&0.043&0.042&0.044&0.049&0.055&&0.050&0.046&0.044&0.041&0.037&0.041&0.039&0.043&0.050\\

\hline
\multicolumn{10}{c}{Panel C: $T=800$, $\delta = -0.50$}&\multicolumn{10}{c}{Panel D: $T=1600$, $\delta = -0.50$}
\\
&\multicolumn{9}{c}{Standard quantile $t$-test} 
&&\multicolumn{9}{c}{Standard quantile $t$-test}  \\ 
0&0.123&0.126&0.130&0.130&0.131&0.124&0.129&0.127&0.128&&0.122&0.130&0.131&0.137&0.140&0.136&0.133&0.130&0.120\\
-5&0.092&0.085&0.089&0.087&0.086&0.083&0.085&0.087&0.095&&0.089&0.092&0.090&0.095&0.091&0.091&0.092&0.091&0.090\\
-10&0.081&0.076&0.078&0.074&0.076&0.071&0.074&0.076&0.082&&0.081&0.082&0.079&0.078&0.076&0.076&0.076&0.078&0.079\\
-25&0.071&0.065&0.068&0.062&0.064&0.060&0.064&0.069&0.072&&0.070&0.069&0.067&0.066&0.066&0.068&0.063&0.067&0.067\\
-200&0.061&0.055&0.054&0.051&0.049&0.051&0.053&0.056&0.058&&0.061&0.057&0.057&0.055&0.051&0.053&0.052&0.055&0.057\\

&\multicolumn{9}{c}{Switching-FM test}  &&\multicolumn{9}{c}{Switching-FM test}\\ 

0&0.076&0.065&0.059&0.056&0.055&0.058&0.059&0.069&0.076&&0.065&0.058&0.057&0.056&0.056&0.056&0.057&0.057&0.066\\
-5&0.072&0.059&0.059&0.055&0.054&0.053&0.054&0.063&0.074&&0.066&0.060&0.055&0.052&0.051&0.053&0.054&0.057&0.065\\
-10&0.068&0.057&0.054&0.053&0.052&0.049&0.055&0.059&0.069&&0.066&0.061&0.053&0.051&0.049&0.052&0.053&0.052&0.062\\
-25&0.061&0.053&0.049&0.047&0.047&0.044&0.050&0.054&0.062&&0.059&0.052&0.051&0.049&0.047&0.049&0.047&0.050&0.055\\
-200&0.059&0.053&0.050&0.046&0.048&0.049&0.051&0.054&0.055&&0.058&0.056&0.053&0.052&0.047&0.052&0.051&0.053&0.056\\

&\multicolumn{9}{c}{$t^w$ test}  &&\multicolumn{9}{c}{$t^w$ test} \\
0&0.070&0.063&0.058&0.059&0.060&0.059&0.062&0.064&0.073&&0.063&0.059&0.055&0.053&0.056&0.056&0.055&0.057&0.065\\
-5&0.053&0.048&0.047&0.045&0.046&0.044&0.046&0.048&0.054&&0.056&0.045&0.043&0.042&0.040&0.044&0.045&0.046&0.054\\
-10&0.053&0.048&0.045&0.042&0.038&0.039&0.040&0.040&0.053&&0.049&0.041&0.040&0.037&0.038&0.040&0.041&0.040&0.047\\
-25&0.053&0.045&0.044&0.039&0.037&0.039&0.042&0.044&0.049&&0.046&0.038&0.039&0.036&0.034&0.038&0.036&0.040&0.044\\
-200&0.056&0.050&0.048&0.046&0.043&0.045&0.045&0.049&0.053&&0.055&0.049&0.047&0.045&0.044&0.042&0.044&0.050&0.054\\
  
\hline
\end{tabular}
\begin{minipage}[c]{ 136  mm}\footnotesize
The table shows null rejection rates for a nominal five percent test of $H_0:\gamma_1(\tau)=0$ against $H_{A}:\gamma_1(\tau)>0$ using a standard quantile regression $t$-test with conventional critical values, the switching-FM predictive quantile test, and the $t^w$ test. We simulate the lagged predictor $x_{t-1}$ from $x_t = (1 - \phi)x_{t-1} + v_t$ with $\phi =  1 + c/T$ and $x_0=0$. We simulate the dependent variable $y_t$ from $y_t = e_t$. The innovations, $(v_{t},e_{t})$, are simulated from an i.i.d.\ bivariate normal distribution with means equal to zero, unit variances, and correlation $\delta$. 
\end{minipage}
\end{table}

\newcommand{\h}[1]{}
\begin{table}
\caption{Null rejection rates of switching-FM test under conditional heteroskedasticity, leverage, and fat-tails  \label{tab:GJR}}
\centering
\newcolumntype{H}{>{\setbox0=\hbox\bgroup}c<{\egroup}@{}}
\begin{tabular}{ccHcHcHcHcccHcHcHcHc}
\hline
$c\backslash\tau $& 0.100&0.200&0.300&0.400&0.500&0.600&0.700&0.800&0.900&&
 0.100&0.200&0.300&0.400&0.500&0.600&0.700&0.800&0.900\\
\hline
%
%
\multicolumn{20}{c}{Panel A: Exogenous GRJ-GARCH-$t$}\\
\multicolumn{10}{c}{$T=800$}&\multicolumn{10}{c}{$T=1600$}\\

0&0.102&0.080&0.067&0.054&0.044&0.040&0.056&0.075&0.097&&0.095&0.083&0.066&0.050&0.039&0.042&0.052&0.073&0.092\\
-5&0.095&0.078&0.065&0.053&0.045&0.046&0.056&0.071&0.099&&0.088&0.075&0.063&0.050&0.040&0.040&0.053&0.071&0.093\\
-10&0.089&0.073&0.060&0.052&0.041&0.045&0.055&0.068&0.092&&0.086&0.070&0.058&0.046&0.039&0.039&0.051&0.067&0.086\\
-25&0.074&0.063&0.052&0.046&0.045&0.046&0.054&0.067&0.079&&0.077&0.064&0.051&0.044&0.040&0.040&0.047&0.060&0.078\\
-200&0.064&0.052&0.044&0.041&0.040&0.041&0.050&0.059&0.070&&0.064&0.052&0.045&0.036&0.038&0.038&0.044&0.048&0.063\\
\hline

\multicolumn{20}{c}{Panel B: Endogenous Student-$t$}\\
\multicolumn{10}{c}{$T=800$}&\multicolumn{10}{c}{$T=1600$}\\

0&0.111&0.074&0.052&0.040&0.036&0.042&0.053&0.075&0.114&&0.089&0.069&0.048&0.039&0.036&0.040&0.049&0.069&0.095\\
-5&0.097&0.065&0.048&0.039&0.035&0.041&0.051&0.067&0.101&&0.087&0.060&0.046&0.038&0.039&0.039&0.046&0.062&0.082\\
-10&0.091&0.061&0.048&0.039&0.035&0.039&0.049&0.063&0.092&&0.079&0.059&0.044&0.038&0.040&0.038&0.045&0.058&0.079\\
-25&0.079&0.054&0.043&0.036&0.033&0.037&0.043&0.058&0.081&&0.068&0.050&0.038&0.033&0.034&0.034&0.043&0.053&0.072\\
-200&0.090&0.070&0.056&0.047&0.044&0.048&0.055&0.074&0.105&&0.081&0.064&0.054&0.044&0.044&0.047&0.056&0.070&0.084\\

\hline
\end{tabular}

\begin{minipage}[c]{
 145
 mm}\footnotesize
The table shows null rejection rates for a nominal five percent test of $H_0:\gamma_1(\tau)=0$ against $H_{A}:\gamma_1(\tau)>0$ using the switching-FM predictive quantile test. We generate the lagged predictor $x_{t-1}$ from $x_t = (1 - \phi)x_{t-1} + v_t$ with $\phi =  1 + c/T$ and $x_0=0$. In Panel A, we simulate $y_t$ from $y_t=e_t$, where $e_t = \sigma_t\ve_{t}$ follows the GJR-GARCH-$t$ process $\sigma_t^2= 0.0001+ 0.0558e_{t-1}^2 +0.1382 I(e_{t-1}<0)e_{t-1}^2+ 0.8226 \sigma_{t-1}^2$, and  $\ve_{t}$ and $v_t$ are drawn from mutually independent and i.i.d.\ Student's $t$-distributions. In Panel B, we simulate $y_t$ from $y_t=e_t$, where $(e_t,v_t)$ is drawn from a multivariate Student's $t$-distribution with correlation $\delta=-0.95$. The innovation degree of freedom parameter is set to $8$ in Panel A and to $3$ in Panel B. All innovations are rescaled to have unit variances. Further details can be found in Section \ref{sec:simulation} of the text.
\end{minipage}
\end{table}


\clearpage

\clearpage
\newcommand{\rcg}{\rowcolor{Gray}}
\newcommand{\ccw}{\cellcolor{White}}

\begin{table}
  \centering
  \caption{Size and power under the traditional linear alternative \label{table:power:n400:high:delta}  }
  \footnotesize
  \begin{tabular}{m{.035 \textwidth}ccccccm{0.000 \textwidth}cccccccc}
\hline
$\tau\backslash\gamma^*$ & 0 & 5 & 10 & 15 & 20 & 25 && 0 & 5 & 10 & 15 & 20 & 25 
\\
\hline

&\multicolumn{6}{c}{Panel A:   
  $T=800$, $c=-5$}
&&\multicolumn{6}{c}{Panel B:   $T=1600$, $c=-5$}
\\

&\multicolumn{6}{c}{Switching-FM test} 
&&\multicolumn{6}{c}{Switching-FM test}  \\ 

\rcg \ccw 0.10&\ccw 0.072&0.202&0.463&0.743&0.908&0.974&\ccw&\ccw 0.062&0.186&0.452&0.745&0.916&0.977\\
\rcg \ccw 0.50&\ccw 0.047&0.225&0.669&0.934&0.991&0.999&\ccw&\ccw 0.051&0.229&0.678&0.949&0.995&1.000\\
\rcg \ccw 0.90&\ccw 0.068&0.207&0.478&0.755&0.912&0.973&\ccw&\ccw 0.060&0.183&0.452&0.738&0.911&0.973\\

&\multicolumn{6}{c}{$t^w$ test} 
&&\multicolumn{6}{c}{$t^w$ test} 
\\
\rcg \ccw 0.10&\ccw 0.051&0.110&0.225&0.371&0.540&0.700&\ccw&\ccw 0.044&0.100&0.188&0.319&0.480&0.618\\
\rcg \ccw 0.50&\ccw 0.040&0.111&0.280&0.509&0.735&0.872&\ccw&\ccw 0.038&0.096&0.236&0.434&0.646&0.811\\
\rcg \ccw 0.90&\ccw 0.052&0.111&0.224&0.378&0.551&0.703&\ccw&\ccw 0.047&0.092&0.189&0.321&0.468&0.614\\
\hline
&\multicolumn{6}{c}{Panel C:  
  $T=800$, $c=-10$} 
&&\multicolumn{6}{c}{Panel D: 
  $T=1600$, $c=-10$} 
\\

&\multicolumn{6}{c}{Switching-FM test}  
&&\multicolumn{6}{c}{Switching-FM test }  \\
\rcg \ccw 0.10&\ccw 0.067&0.158&0.325&0.567&0.783&0.917&\ccw&\ccw 0.056&0.147&0.321&0.562&0.786&0.922\\
\rcg \ccw 0.50&\ccw 0.050&0.165&0.446&0.801&0.960&0.993&\ccw&\ccw 0.051&0.164&0.453&0.814&0.970&0.997\\
\rcg \ccw 0.90&\ccw 0.065&0.160&0.339&0.580&0.788&0.916&\ccw&\ccw 0.055&0.141&0.317&0.556&0.781&0.917\\

&\multicolumn{6}{c}{$t^w$ test} 
&&\multicolumn{6}{c}{$t^w$ test} 
\\

\rcg \ccw 0.10&\ccw 0.048&0.092&0.179&0.286&0.426&0.581&\ccw&\ccw 0.041&0.084&0.141&0.236&0.361&0.499\\
\rcg \ccw 0.50&\ccw 0.037&0.095&0.199&0.383&0.593&0.782&\ccw&\ccw 0.029&0.076&0.162&0.314&0.497&0.685\\
\rcg \ccw 0.90&\ccw 0.048&0.093&0.172&0.286&0.436&0.593&\ccw&\ccw 0.040&0.080&0.142&0.237&0.358&0.502\\
\hline

&\multicolumn{6}{c}{Panel E: 
  $T=800$ $c=-25$} 
&&\multicolumn{6}{c}{Panel F: $T=1600$, $c=-25$} \\

&\multicolumn{6}{c}{Switching-FM test}   
&&\multicolumn{6}{c}{Switching-FM test}  \\ 

\rcg \ccw 0.10&\ccw 0.058&0.109&0.188&0.308&0.468&0.639&\ccw&\ccw 0.049&0.100&0.189&0.311&0.471&0.636\\
\rcg \ccw 0.50&\ccw 0.041&0.102&0.223&0.421&0.661&0.851&\ccw&\ccw 0.038&0.103&0.226&0.433&0.684&0.874\\
\rcg \ccw 0.90&\ccw 0.057&0.107&0.196&0.325&0.482&0.643&\ccw&\ccw 0.048&0.095&0.176&0.298&0.461&0.632\\

&\multicolumn{6}{c}{$t^w$ test} 
&&\multicolumn{6}{c}{$t^w$ test} 
\\
\rcg \ccw 0.10&\ccw 0.046&0.078&0.129&0.202&0.291&0.400&\ccw&\ccw 0.040&0.067&0.109&0.170&0.249&0.345\\
\rcg \ccw 0.50&\ccw 0.033&0.073&0.143&0.246&0.395&0.569&\ccw&\ccw 0.032&0.065&0.118&0.209&0.332&0.466\\
\rcg \ccw 0.90&\ccw 0.045&0.082&0.137&0.199&0.300&0.403&\ccw&\ccw 0.040&0.069&0.111&0.166&0.245&0.342\\

\hline
\end{tabular}
\begin{minipage}[c]{ 175  mm}\footnotesize
The table shows rejection rates under both the null (Columns 2 \& 8, unshaded) and alternative (Columns 3--7 \& 9--13, shaded grey) hypotheses for a nominal five percent test of  $H_0:\gamma_1(\tau)=0$ against $H_{A}:\gamma_1(\tau)>0$, for both the switching-FM and $t^w$ predictive quantile tests. We simulate the lagged predictor $x_{t-1}$ from $x_t = (1 - \phi)x_{t-1} + v_t$ with $\phi =  1 + c/T$ and $x_0=0$. We generate $y_t$ from $y_t =(\gamma_1^*/T) x_{t-1} + e_t$. The innovations, $(v_{t},e_{t})$, are drawn from an i.i.d.\ bivariate normal distribution with means equal to zero, unit variances, and correlation $\delta=-0.95$.
\end{minipage}
  \end{table}


\begin{table}
  \caption{Size and power in a random coefficient model with tail predictability 
  \label{table:power:tail:high:delta}  }
\footnotesize
  \begin{tabular}{lccccccm{0.02 \textwidth}ccccccc}
\hline
$\tau\backslash b$ & 0 & 2.5 & 5 &7.5 & 10 & 12.5 & & 0 & 2.5 & 5 &7.5 & 10 & 12.5 \\
\hline
& \multicolumn{6}{c}{Panel A: $T=800$, $c=-5$} 
& &\multicolumn{6}{c}{Panel B: $T=1600$, $c=-5$} &\\
& \multicolumn{6}{c}{Switching-FM test} & &\multicolumn{6}{c}{Switching-FM test} & \\ 

0.50&0.047&0.049&0.048&0.049&0.048&0.048&&0.051&0.051&0.049&0.050&0.051&0.050\\
\rcg \ccw 0.70&\ccw 0.052&0.219&0.432&0.597&0.695&0.753&\ccw &\ccw 0.048&0.323&0.691&0.867&0.931&0.956\\
\rcg \ccw 0.90&\ccw 0.068&0.602&0.883&0.953&0.971&0.981&\ccw &\ccw 0.060&0.766&0.978&0.995&0.998&0.999\\

&\multicolumn{6}{c}{$t^w$ test} 
& & \multicolumn{6}{c}{$t^w$ test} &\\ 

0.50&0.039&0.041&0.044&0.047&0.046&0.042&&0.035&0.037&0.035&0.037&0.041&0.039\\
\rcg \ccw 0.70&\ccw 0.041&0.111&0.191&0.254&0.306&0.344&\ccw &\ccw 0.036&0.128&0.257&0.383&0.502&0.576\\
\rcg \ccw 0.90&\ccw 0.054&0.283&0.521&0.656&0.744&0.781&\ccw &\ccw 0.044&0.351&0.663&0.828&0.898&0.938\\



\\

\hline
&\multicolumn{6}{c}{Panel C: $T=800$, $c=-10$} 
& &\multicolumn{6}{c}{Panel D: $T=1600$, $c=-10$} &\\
&\multicolumn{6}{c}{Switching-FM test} 
& &\multicolumn{6}{c}{Switching-FM test} & \\ 

0.50&0.050&0.048&0.048&0.047&0.048&0.048&&0.051&0.051&0.048&0.048&0.047&0.046\\
\rcg \ccw 0.70&\ccw 0.049&0.161&0.285&0.394&0.480&0.544&\ccw &\ccw 0.046&0.230&0.508&0.736&0.855&0.912\\
\rcg \ccw 0.90&\ccw 0.065&0.435&0.743&0.873&0.929&0.951&\ccw &\ccw 0.055&0.617&0.947&0.991&0.998&0.999\\

&\multicolumn{6}{c}{$t^w$ test} 
& & \multicolumn{6}{c}{$t^w$ test} &\\

0.50&0.038&0.036&0.041&0.038&0.039&0.037&&0.030&0.029&0.030&0.032&0.032&0.034\\
\rcg \ccw 0.70&\ccw 0.040&0.089&0.143&0.192&0.222&0.252&\ccw &\ccw 0.033&0.098&0.185&0.292&0.369&0.431\\
\rcg \ccw 0.90&\ccw 0.047&0.217&0.400&0.536&0.622&0.677&\ccw &\ccw 0.042&0.265&0.552&0.753&0.858&0.904\\

\hline
&\multicolumn{6}{c}{Panel E: $T=800$, $c=-25$} 
& &\multicolumn{6}{c}{Panel F:  $T=1600$, $c=-25$} &\\
&\multicolumn{6}{c}{Switching-FM test} 
&&\multicolumn{6}{c}{Switching-FM test} & \\ 
 
0.50&0.041&0.040&0.041&0.041&0.040&0.040&&0.038&0.039&0.039&0.038&0.038&0.038\\
\rcg \ccw 0.70&\ccw 0.043&0.098&0.157&0.203&0.239&0.265&\ccw &\ccw 0.040&0.131&0.255&0.380&0.491&0.581\\
\rcg \ccw 0.90&\ccw 0.057&0.244&0.438&0.577&0.667&0.728&\ccw &\ccw 0.048&0.337&0.703&0.886&0.956&0.982\\

&\multicolumn{6}{c}{$t^w$ test} & & \multicolumn{6}{c}{$t^w$ test} &\\ 

0.50&0.034&0.032&0.033&0.035&0.031&0.035&&0.031&0.028&0.031&0.031&0.030&0.029\\
\rcg \ccw 0.70&\ccw 0.035&0.073&0.113&0.130&0.155&0.171&\ccw &\ccw 0.031&0.079&0.135&0.192&0.236&0.287\\
\rcg \ccw 0.90&\ccw 0.048&0.159&0.275&0.359&0.421&0.481&\ccw &\ccw 0.038&0.185&0.386&0.550&0.671&0.754\\

\hline
\end{tabular}
\begin{minipage}[c]{ 175  mm}\footnotesize
  The table shows rejection rates for a nominal five percent test of $H_0:\gamma_1(\tau)=0$, against $H_{A}:\gamma_1(\tau)>0$. We simulate the lagged predictor $x_{t-1}$ from $x_t = (1 - \phi)x_{t-1} + v_t,$ with $\phi =  1 + c/T$ and $x_0=0$. $y_t$ is generated by $y_t = e_t  + T^{-0.75}be_t |x_{t-1}+25|$. The innovations, $(v_{t},e_{t})$, are drawn from an i.i.d.\ bivariate normal distribution with means equal to zero, unit variances, and correlation $\delta=-0.95$. This specification allows for predictability in the tail, but not the center, of the distribution. The null hypothesis holds when either  $b=0$ and/or $\tau = 0.5$ (unshaded region). Rejection rates under the alternative hypothesis ($b>0$ and $\tau>0.5$) are shown in the grey-shaded region.
\end{minipage}
  \end{table}

\begin{table}
  \centering
  \caption{Size and power in a random coefficient model with increasing predictability at higher quantiles \label{table:power:randcoef:cai}  }
  \footnotesize
  \newcolumntype{H}{>{\setbox0=\hbox\bgroup}c<{\egroup}@{}}
  \begin{tabular}{m{.035 \textwidth}ccccccHHHHHm{0.000 \textwidth}ccccccHHHHH}
\hline
$\tau\backslash\gamma^*$ &
0&2.5&5&7.5&10&12.5&15&17.5&20&22.5&25&&
0&2.5&5&7.5&10&12.5&15&17.5&20&22.5&25\\

\\
\hline

&\multicolumn{11}{c}{Panel A: $T=800$, $c=-5$}
&&\multicolumn{11}{c}{Panel B: $T=1600$, $c=-5$}
\\

&\multicolumn{11}{c}{Switching-FM test} 
&&\multicolumn{11}{c}{Switching-FM test}  \\ 
\rcg \ccw 0.10&\ccw 0.072&0.176&0.376&0.645&0.837&0.934&0.976&0.991&0.996&0.998&0.999&\ccw&\ccw  0.062&0.161&0.370&0.632&0.843&0.941&0.979&0.994&0.998&1.000&1.000\\
\rcg \ccw 0.50&\ccw 0.047&0.431&0.936&0.995&1.000&1.000&1.000&1.000&1.000&1.000&1.000&\ccw&\ccw  0.051&0.433&0.950&0.998&1.000&1.000&1.000&1.000&1.000&1.000&1.000\\
\rcg \ccw 0.90&\ccw 0.068&0.522&0.936&0.995&1.000&1.000&1.000&1.000&1.000&1.000&1.000&\ccw&\ccw  0.060&0.494&0.934&0.997&1.000&1.000&1.000&1.000&1.000&1.000&1.000\\

&\multicolumn{11}{c}{$t^w$ test} 
&&\multicolumn{11}{c}{$t^w$ test} 
\\
\rcg \ccw 0.10&\ccw 0.054&0.099&0.185&0.313&0.460&0.603&0.732&0.822&0.893&0.933&0.955&\ccw&\ccw  0.047&0.088&0.153&0.258&0.393&0.533&0.655&0.758&0.833&0.885&0.926\\
\rcg \ccw 0.50&\ccw 0.039&0.178&0.513&0.813&0.945&0.981&0.992&0.996&0.997&0.998&0.999&\ccw&\ccw  0.035&0.154&0.434&0.736&0.900&0.969&0.988&0.992&0.996&0.998&0.998\\
\rcg \ccw 0.90&\ccw 0.054&0.244&0.603&0.859&0.957&0.984&0.995&0.997&0.997&0.999&0.999&\ccw&\ccw  0.044&0.199&0.511&0.788&0.924&0.973&0.989&0.994&0.996&0.998&0.998\\

\hline
&\multicolumn{11}{c}{Panel C: $T=800$, $c=-10$} 
&&\multicolumn{11}{c}{Panel D: $T=1600$, $c=-10$} 
\\

&\multicolumn{11}{c}{Switching-FM test}  
&&\multicolumn{11}{c}{Switching-FM test}  \\

\rcg \ccw 0.10&\ccw 0.067&0.140&0.273&0.461&0.680&0.840&0.934&0.972&0.989&0.996&0.998&\ccw&\ccw  0.056&0.131&0.264&0.455&0.673&0.843&0.936&0.976&0.993&0.998&1.000\\
\rcg \ccw 0.50&\ccw 0.050&0.282&0.804&0.982&0.999&1.000&1.000&1.000&1.000&1.000&1.000&\ccw&\ccw  0.051&0.284&0.814&0.989&1.000&1.000&1.000&1.000&1.000&1.000&1.000\\
\rcg \ccw 0.90&\ccw 0.065&0.369&0.833&0.983&0.999&1.000&1.000&1.000&1.000&1.000&1.000&\ccw&\ccw  0.055&0.346&0.829&0.985&0.999&1.000&1.000&1.000&1.000&1.000&1.000\\

&\multicolumn{11}{c}{$t^w$ test}
&&\multicolumn{11}{c}{$t^w$ test}
\\
\rcg \ccw 0.10&\ccw 0.045&0.087&0.150&0.240&0.351&0.486&0.609&0.727&0.821&0.894&0.935&\ccw&\ccw  0.042&0.076&0.119&0.199&0.292&0.402&0.526&0.641&0.747&0.824&0.887\\
\rcg \ccw 0.50&\ccw 0.038&0.138&0.381&0.696&0.905&0.980&0.997&0.999&1.000&1.000&1.000&\ccw&\ccw  0.030&0.112&0.302&0.590&0.828&0.952&0.988&0.998&0.999&1.000&1.000\\
\rcg \ccw 0.90&\ccw 0.047&0.188&0.481&0.779&0.935&0.983&0.996&0.999&1.000&1.000&1.000&\ccw&\ccw  0.042&0.153&0.394&0.683&0.876&0.960&0.991&0.997&0.999&1.000&1.000\\


\hline

&\multicolumn{11}{c}{Panel E: 
  $T=800$ $c=-25$} 
&&\multicolumn{11}{c}{Panel F: $T=1600$, $c=-25$} \\

&\multicolumn{11}{c}{Switching-FM test}
&&\multicolumn{11}{c}{Switching-FM test}  \\ 

\rcg \ccw 0.10&\ccw 0.058&0.099&0.162&0.256&0.377&0.524&0.673&0.790&0.885&0.942&0.972&\ccw&\ccw  0.049&0.091&0.159&0.254&0.383&0.523&0.667&0.793&0.887&0.942&0.975\\
\rcg \ccw 0.50&\ccw 0.041&0.155&0.421&0.770&0.957&0.995&0.999&1.000&1.000&1.000&1.000&\ccw&\ccw  0.038&0.155&0.434&0.789&0.967&0.998&1.000&1.000&1.000&1.000&1.000\\
\rcg \ccw 0.90&\ccw 0.057&0.212&0.531&0.830&0.966&0.997&1.000&1.000&1.000&1.000&1.000&\ccw&\ccw  0.048&0.190&0.509&0.831&0.969&0.998&1.000&1.000&1.000&1.000&1.000\\

&\multicolumn{11}{c}{$t^w$ test} 
&&\multicolumn{11}{c}{$t^w$ test} 
\\
\rcg \ccw 0.10&\ccw 0.046&0.076&0.116&0.169&0.245&0.329&0.427&0.527&0.637&0.721&0.798&\ccw&\ccw  0.038&0.064&0.098&0.148&0.205&0.280&0.359&0.455&0.544&0.638&0.716\\
\rcg \ccw 0.50&\ccw 0.034&0.105&0.252&0.474&0.709&0.878&0.966&0.994&0.999&1.000&1.000&\ccw&\ccw  0.031&0.088&0.211&0.400&0.612&0.796&0.913&0.973&0.995&0.999&1.000\\
\rcg \ccw 0.90&\ccw 0.048&0.140&0.334&0.576&0.789&0.922&0.977&0.996&0.999&1.000&1.000&\ccw&\ccw  0.038&0.114&0.274&0.486&0.704&0.864&0.946&0.983&0.996&1.000&1.000\\

\hline
\end{tabular}
\begin{minipage}[c]{ 175  mm}\footnotesize
The table shows rejection rates under both the null (Columns 2 \& 8, unshaded) and alternative (Columns 3--7 \& 9--13, shaded grey) hypotheses for a nominal five percent test of  $H_0:\gamma_1(\tau)=0$ against $H_{A}:\gamma_1(\tau)>0$, for both the switching-FM and $t^w$ predictive quantile tests. We generate the lagged predictor $x_{t-1}$ from $x_t = (1 - \phi)x_{t-1} + v_t$ with $\phi =  1 + c/T$ and $x_0=0$. We simulate $y_t$ from $y_t =3(1 + T^{-1}\gamma^* x_{t-1} )+  (1 + T^{-1}\gamma^* x_{t-1} ) e_t$. The innovations, $(v_{t},e_{t})$, are drawn from an i.i.d.\ bivariate normal distribution with means equal to zero, unit variances, and correlation $\delta=-0.95$. The null hypothesis is imposed by $\gamma_1^*= 0$ and the alternative is given by $\gamma_1^*>0$. 
\end{minipage}
  \end{table}


\clearpage
\begin{table}
\caption{Preliminary indications of predictive regression problem\label{table:indicators} }
\begin{center}
\begin{tabular}{lrrr}
predictor: & $\mbox{dp}_{t-1}$ & $\mbox{ep}_{t-1}$ & $\mbox{bm}_{t-1}$ \\\hline
$t_{\mbox{\tiny{DFGLS}}}$ & -1.4485 & -2.3258 & -1.7987 \\
$\widehat{\phi}_{\mbox{\tiny{DFGLS}}}$ & 0.9958 & 0.9928 & 0.9939 \\
$c_L^{0.95}$ & -10.6590 & -20.3347 & -14.1414 \\
$c_U^{0.95}$ & 3.0068 & -0.5643 & 1.8878 \\
$\phi_{L}$ & 0.9901 & 0.9811 & 0.9869 \\
$\phi_{U}$ & 1.0028 & 0.9995 & 1.0018 \\
$\widehat{\delta} $ & -0.9738 & -0.7971 & -0.7847 \\
\hline
\end{tabular}
\begin{minipage}[c]{ 125  mm}\footnotesize
The table shows preliminary indications of the predictor persistence and endogeneity. Rows 2-3 provide the DFGLS $t$-statistic and parameter estimate using intercept only. The 5\% critical value is $-1.95$. Given in Rows 4--5, $(c_L^{0.95},c_U^{0.95})$ provides the resulting 95\% confidence interval on $c$. $(\phi_L,\phi_U)$, in Rows 6--7, restates this interval in terms of $\phi$. $\widehat{\delta}$ in Row 8 estimates the contemporaneous correlation between the return and predictor innovations.
\end{minipage}
\end{center}
\end{table}

\begin{table}
\caption{Results from standard, HAC, and switching-FM predictive quantile regression tests\label{table:fm-switch} }
\begin{center}
\begin{tabular}{lrrrrrrrrrr}
\hline
quantile &0.1&0.2&0.3&0.4&0.5&0.6&0.7&0.8&0.9\\
\hline
\multicolumn{10}{c}{log dividend price ratio} \\
$\widehat{\gamma}_1(\tau)$ & -0.0014&-0.0003& 0.0001& 0.0002& 0.0003& 0.0008& 0.0014& 0.0015& 0.0019 \\
$t_{\gamma_1}^{std}(\tau)$&\textbf{-2.0819}&-0.6663&0.1702&0.6217&0.8239&\textbf{2.1625}&\textbf{3.3109}&\textbf{3.9472}&\textbf{3.2400} \\
$t_{\gamma_1}(\tau)$ & -0.7523&-0.3590& 0.0878& 0.3477& 0.5754& \textbf{1.8738}& \textbf{2.4840}&\textbf{2.1136}& \textbf{1.8757} \\
$\widehat{\delta}_\tau$ & -0.6458&-0.5412&-0.5375&-0.6442&-0.6178&-0.6194&-0.5957&-0.5554&-0.5223 \\
$c_L$ & -8.7342&-8.3649&-8.3557&-8.7214&-8.5975&-8.6035&-8.5167&-8.4004&-8.3175\\
$c_U$ & -1.7803&-2.2027&-2.2153&-1.7842&-1.9041&-1.8957&-2.0148&-2.1544&-2.2673 \\
$\underline{\gamma}_1(\tau)$ & -0.0046&-0.0019&-0.0012&-0.0009&-0.0007&-0.0002& \textbf{0.0002}& \textbf{0.0002}&-0.0004\\
$\overline{\gamma}_1(\tau)$  & 0.0001& 0.0006& 0.0008& 0.0008& 0.0008& 0.0014& 0.0020& 0.0022& 0.0032 \\
\hline

\multicolumn{10}{c}{log earnings price ratio}\\
$\widehat{\gamma}_1(\tau)$ & 0.0031& 0.0007& 0.0009& 0.0007& 0.0006& 0.0010& 0.0010& 0.0008& 0.0005\\
$t_{\gamma_1}^{std}(\tau)$&\textbf{3.1236}&1.0808&1.5081&1.4422&1.1967&\textbf{1.6780}&1.6258&1.3132&0.4680 \\
$t_{\gamma_1}(\tau)$ & \textbf{2.3848}& 0.7153& 0.9120& 0.8364& 0.6569& 1.2421& 0.9009& 0.6590& 0.3906\\
$\widehat{\delta}_\tau$ & -0.0399&-0.1108&-0.1637&-0.3031&-0.3385&-0.3275&-0.3440&-0.2737&-0.2753\\
$c_L$ & -11.5372&-14.1946&-15.1852&-16.2020&-16.3657&-16.3310&-16.3833&-16.0006&-16.0149\\
$c_U$ & -10.0860&-10.0860&-10.0860&-9.6421&-9.3354&-9.4666&-9.2690&-9.7972&-9.7838\\
$\underline{\gamma}_1(\tau)$ & -0.0003&-0.0015&-0.0009&-0.0010&-0.0011&-0.0007&-0.0008&-0.0012&-0.0028\\
$\overline{\gamma}_1(\tau)$ & 0.0062& 0.0024& 0.0021& 0.0016& 0.0015& 0.0019& 0.0019& 0.0020& 0.0024\\
\hline

\multicolumn{10}{c}{log book-to-market ratio}\\
$\widehat{\gamma}_1(\tau)$ & 0.0000&-0.0001&-0.0001&-0.0001&-0.0000& 0.0003& 0.0009& 0.0008& 0.0013 \\
$t_{\gamma_1}^{std}(\tau)$&0.0308&-0.3068&-0.2892&-0.4091&-0.0017&0.9693&\textbf{2.1117}&\textbf{2.3225}&\textbf{2.1094} \\
$t_{\gamma_1}(\tau)$ & 0.0152&-0.2077&-0.1781&-0.2170&-0.0011& 0.6894& 1.3843& 0.9307& 1.3601\\
$\widehat{\delta}_\tau$ & -0.5249&-0.4488&-0.5034&-0.5548&-0.5132&-0.5107&-0.4782&-0.4735&-0.4021\\
$c_L$ & -11.4494&-11.2051&-11.3821&-11.5481&-11.4120&-11.4045&-11.3049&-11.2892&-11.0625\\
$c_U$ & -4.2316&-4.5181&-4.3005&-4.1266&-4.2705&-4.2781&-4.3780&-4.3999&-4.6866\\
$\underline{\gamma}_1(\tau)$ & -0.0026&-0.0014&-0.0012&-0.0010&-0.0008&-0.0005&-0.0000&-0.0003&-0.0006\\
$\overline{\gamma}_1(\tau)$ & 0.0023& 0.0011& 0.0008& 0.0007& 0.0007& 0.0011& 0.0017& 0.0018& 0.0030\\
\hline
\end{tabular}
\begin{minipage}[c]{ 175 mm}\footnotesize
The table shows empirical results for the standard, HAC, and switching-FM predictive quantile regression tests. Within each panel, the first two rows give the slope estimate ($\widehat{\gamma}_1(\tau)$) and $t$-statistic ($t_{\gamma_1}^{std}(\tau)$) in a standard quantile regression. Row 3 provides the HAC $t$-statistic ($t_{\gamma_1}(\tau)$). In Row 4, we estimate the long-run residual cross-correlation $\delta_{\tau}$. Rows 5--6 provide the adjusted confidence interval for $c$, $(c_L,c_U)$, using the confidence level $\underline{\alpha}_1$ from Table \ref{table:alpha:1:A} corresponding to $\widehat{\delta}_{\tau}$. The final two rows give the resulting fully modified lower and upper bounds ($\underline{\gamma}_1,\overline{\gamma}_1$) for the quantile regression slope coefficient $\gamma_1$. Because the confidence intervals $(c_L,c_U)$ in Rows 5--6 are above $\underline c_L$, the switching-FM quantile regresstion tests rejects in favor of $H_A:\gamma(\tau)>0$ ($H_A:\gamma(\tau)<0$) when $\underline{\gamma}_1>0$ ($\overline{\gamma}_1<0$). Results that are significant at the 5 percent level are marked in bold. 
\end{minipage}
\end{center}
\end{table}

\end{appendices}

\end{document}